\documentclass[a4paper,oneside]{article}

\usepackage{amsmath}%
\usepackage{amssymb}%
\usepackage[latin1]{inputenc}
\usepackage[numbers]{natbib}

\usepackage{geometry}
\newtheorem{theorem}{Theorem}[section]

\newtheorem{corollary}[theorem]{Corollary}

\newtheorem{definition}[theorem]{Definition}

\newtheorem{lemma}[theorem]{Lemma}

\newtheorem{proposition}[theorem]{Proposition}
\newtheorem{remark}[theorem]{Remark}

\newenvironment{proof}[1][Proof]{\textbf{#1.} }{\ \rule{0.5em}{0.5em}}

\newcommand{\refeqn}[1]{(\ref{#1})}

\newcommand{\cinf}[0]{C^{\infty}}

\newcommand{\spann}[0]{\operatorname{span}}

\newcommand{\jinf}[0]{J^{\infty}}

\newcommand{\acosh}[0]{\operatorname{acosh}}

\begin{document}

\title{{\bf The geometry of differential constraints for a class of evolution PDEs}}
\author{ Francesco C. De Vecchi* and Paola Morando** \\
*Dip. di Matematica, Universit\`a degli Studi di Milano, \\
via Saldini 50, Milano;\\
**DISAA, Universit\`a degli Studi di Milano, \\
via Celoria 2, Milano;\\
francesco.devecchi@unimi.it, paola.morando@unimi.it}
\date{}

\maketitle



\begin{abstract}
The problem of computing differential constraints for a family of evolution PDEs is discussed from a constructive point of view. A new method, based on  the existence of generalized characteristics for evolution vector fields, is proposed in order  to obtain explicit differential constraints for  PDEs belonging to this family. Several examples, with applications in non-linear stochastic filtering theory, stochastic perturbation of soliton equations and non-isospectral integrable systems, are discussed in detail to verify the  effectiveness  of the method.
\end{abstract}

\section{Introduction}

The method of differential constraints is a well known and general method for determining particular explicit solutions to a partial differential equation (PDE) reducing the PDE to a system of  ordinary differential equations (ODEs) through a suitable
ansatz on the form of the solution. In particular, given a   system of  evolution equations of the form
$\partial_t(u^k)=F^k(x,t,u,u_{\sigma})$, where $(x^i, u^k)\in M \times N$
and $u^k_{\sigma}$ are the derivatives of $u^k$ with respect to $x^j$ the number of times defined by the multi-index $\sigma$,  we can look for  solutions of the form
\begin{equation}\label{equation_K}
u(x,t)=K(x,w^1(z),...,w^L(z)),
\end{equation}
where  $K:M \times \mathbb{R}^L \to N$ and $z: M \times \mathbb{R} \to  \mathbb{R}$ are smooth functions. Replacing this ansatz in the initial evolution equation, we may
obtain a system of ODEs for the functions $w^i$ with respect to the variable $z$.
We remark that, for general  functions $K$ and $z$,  the system of ODEs for $w^i$ is overdetermined and  has no solutions. When the system
for the functions $w^i$ admits solutions, the ansatz \eqref{equation_K} is said compatible with the equation and  $K$ is called a \emph{differential constraint} for $\partial_t(u^k)=F^k(x,t,u,u_{\sigma})$.
This method appears with different names in several  papers and books
(see, e.g.,   \cite{Clarkson_Kruskal,Galaktionov,Olver2,Kruglikov,Meleshko,Olver1,Pucci_Saccomandi}) and is equivalent to append to the original equation a suitable   overdetermined  system of  PDEs of
the form $\mathcal{I}=\{I_1(x,t,u,u_{\sigma})=0,...,I_K(x,t,u,u_{\sigma})=0\}$.
Indeed, if $\mathcal{I}$ admits a finite dimensional solution, this can be described by a function of the form \refeqn{equation_K} and  the requirement that $\mathcal{I}$ and $\partial_t(u^k)=F^k(x,t,u,u_{\sigma})$ have common solutions can be interpreted as the compatibility condition for the ansatz.
In that situation  the system  $\mathcal{I}$  is called
 differential constraint as well and, in    order to distinguish between the two approaches,  some Authors
refer to \emph{direct differential constraints} for the formulation with the function $K$  and to \emph{indirect differential constraints} when the system   $\mathcal{I}$ is considered (see, e.g. \cite{Olver1,Pucci_Saccomandi}).\\
In general,  checking  that the ansatz \refeqn{equation_K} or the overdetermined system $\mathcal{I}$ are compatible with the evolution equations $\partial_t(u^k)=F^k(x,t,u,u_{\sigma})$  is a difficult task. Even in the simplest case $M=\mathbb{R}$,
in order to verify  that $\mathcal{I}$ is a differential constraint, we have to solve a system
of strongly non-linear PDEs for the unknown  functions $I_k$.
For this reason many Authors look for differential constraints imposing some restrictions on the form of the functions $K$ and $z$, or, equivalently, on the
form of the constraints $I_k$ or, finally, on the form of the equation $\partial_t(u^k)=F^k(x,t,u,u_{\sigma})$.\\

In this paper we are interested in evolution equations with
 a particular form of the functions $F^k$ and  we make suitable assumptions on the form of the differential constraints.  In particular we deal with the problem of finding differential constraints for evolution PDEs  of the form
\begin{equation}\label{equation_problem}
\partial_t (u^k)=\sum_{i=1}^s c^i(t)F^k_i(x,u,u_{\sigma}),
\end{equation}
where the functions $F^k_i$ do not depend on $t$. Furthermore we chose the new  variable  $z(x,t)=t$ and we look for  differential constraints  that are independent
of any possible choice of the smooth functions $c^i(t)$. \\
The choice of equations of the form \refeqn{equation_problem} is not a matter of computational convenience, but is triggered by many theoretical and applied problems arising in different  branches of mathematics.\\
First of all, evolution equations of the form \refeqn{equation_problem} appear in the theory of stochastic processes and in particular in the study of  finite dimensional solutions to stochastic partial differential equations (SPDEs). This topic is discussed in many applications of SPDEs: see,  for examples,  \cite{DeLara1,DeLara2,Hu_Yau_Chiou} for applications to non-linear filtering problems,
\cite{Filipovic1,Filipovic3,Filipovic2} for applications to financial problems and  \cite{Holm1,Wadati,Xie} for applications to stochastic soliton equations.\\
In this framework the problem of finding finite dimensional solutions can be reduced to the problem of finding differential constraints for evolution equations
of the form \refeqn{equation_problem} for any choice of $c^k(t)$. Indeed  we can (formally) see the functions $c^k(t)$ as the
derivatives\footnote{The usual stochastic process considered in SPDE theory is Brownian motion that does not admit derivative almost surely. For this
reason we say that the deterministic problem proposed in this paper is formally equivalent to  the stochastic one. A detailed discussion of the possible application of differential constraints in SPDE theory  will be provided  in a following paper exploiting the geometric techniques developed here.} of some stochastic process which can assume any possible values.\\
Moreover, equations of the form \refeqn{equation_problem} have  applications to integrable system theory, and our results turn out to be useful for dealing with  non-isospectral deformations of integrable systems (see \cite{Berezansky,Calogero,Gordoa_Pickering}). In fact the local non-isospectral deformation of KdV hierarchy can be reduced to the usual isospectral KdV hierarchy by means of a time dependent transformation (see e.g. \cite{Gordoa_Pickering_Wattis} and the
examples of Subsections \ref{subsection_transport_equation} and \ref{subsection_KdV}).\\
Finally, evolution equations of the form \eqref{equation_problem}  can be considered in infinite dimensional control theory (see \cite{Daprato,Lions}). In this framework our method can be interpreted as an application of usual methods of geometric control theory  to the explicit computation of the reachable sets of some particular point (see \cite{Agrachev} for the finite dimensional case and \cite{Khajeh} for the infinite dimensional one).\\
It is worth to remark that, although the problem of finding differential constraints for equations of the form \refeqn{equation_problem} has been faced many times, to the best of our knowledge this is the first time that the abstract form \refeqn{equation_problem} of the problem has been  recognized (since the previous Authors consider particular forms of the functions $F^i_k$) and that  the problem has been tackled  by using the geometrical framework of differential constraints for PDEs. In particular,  the use of the differential constraints method permits to consider, from a theoretical point of view, a very general form for the functions  $F^k_i$  and allows us to obtain a useful  algorithm  for the explicit computation of the solutions to  equations \eqref{equation_problem} (see Subsection \ref{subsection_general_algorithm}).\\

On the other hand, the problem of finding differential constraints for  PDEs of the form  \eqref{equation_problem}  is an interesting challenge in itself.
In order to address this issue we provide a geometrical framework for the description  of differential constraints method which allows us to simplify the formulation of the problem. In particular, we associate with any autonomous
evolution equation of the form $\partial_t(u^k)=F^k_i(x,u,u_{\sigma})$ an evolution vector field $V_{F_i}$ on the space of infinite jets $\jinf(M,N)$ of the
functions from $M$ into $N$. In this setting the time independent overdetermined system $\mathcal{I}$, or equivalently the time independent function $K$, can be described
as a particular finite dimensional submanifold $\mathcal{K}$ of $\jinf(M,N)$ and  we prove that $\mathcal{I}$ is a differential
constraint for the system \refeqn{equation_problem} for any  $c^i(t)$ if and only if   $V_{F_i} \in T\mathcal{K}$ $\forall i=1, \ldots s$. \\
This geometrical reformulation gives new insight into the problem of finding differential constraints. First of all it provides  a general and powerful  method for dealing  with many different evolution equations  which have been previously faced with different techniques and in different frameworks.\\
The second important result following by our general approach is the derivation of some necessary conditions for the existence of differential constraints for systems of the form \eqref{equation_problem}. These conditions form an infinite dimensional analogous of necessary condition of the well-known Frobenius theorem.
Indeed the existence of a differential constraint  for \refeqn{equation_problem}  ensures that  the
functions $F_i$, restricted on some subset $\mathcal{K}$ of $\jinf(M,N)$, form a module with respect to the Lie brackets $[F_i,F_j]$ induced by the vector fields Lie brackets $[V_{F_i},V_{F_j}]$. This imposes severe conditions on  $F_i$:
in particular, if $\mathcal{K}=\jinf(M,N)$,  $F_i$ must form a Lie algebra on $\jinf(M,N)$.  \\
Moreover, in order to obtain a sufficient condition for the existence of differential constraints for systems of the form \refeqn{equation_problem}, we introduce the notion of characteristic flow for a general evolution vector field. In particular this definition, generalizing to  higher order the standard notion of characteristic of a first order scalar evolution equation, is an infinity dimensional  analogous  of the flow of a vector field in finite dimensional framework. Thereafter we divide the functions $F_i$ in two sets $H_i$ and $G_j$ so that the functions $G_j$ form a Lie algebra and their evolution vector fields $V_{G_j}$
admit   generalized characteristic flow.  Under these assumptions we prove that, if the vector fields  $V_{H_i}$ admit a differential constraint $\mathcal{H}$,
then the complete set of  $V_{F_i}$ admits a differential constraint $\mathcal{K}$ that can be explicitly computed starting from $\mathcal{H}$  and using the
characteristic flows of $G_i$.
In addition  we provide a generalization of this theorem to the case of $H_i,G_i$ forming a finite dimensional  Lie algebra on a real analytic  submanifold $\mathcal{H}$ of $\jinf(M,N)$   which is also a differential constraint for $V_{H_i}$.\\
The  proofs of both these results  are constructive so that we can compute explicitly the differential constraints in many interesting examples.
 In particular some of the examples have been chosen in order to show the  flexibility and the effectiveness  of our geometrical approach with respect to the standard differential constraints method (see Subsections \ref{subsection_transport_equation}, \ref{subsection_integrable} and also \cite{DM2016} for other examples on the same topics). Other examples instead have been proposed for their relevance in applied mathematical problems such as non-linear stochastic filtering theory (see Subsection \ref{subsection_heat_equation}), stochastic perturbation of integrable equations (see Subsections \ref{subsection_integrable} and \ref{subsection_KdV}) and non-isospectral deformation of integrable systems (see Subsection \ref{subsection_transport_equation} and \ref{subsection_KdV}).\\

The paper is organized as follows:  after recalling some basic facts on the geometry of $\jinf(M,N)$ in Section \ref{prel},   in Section \ref{PDEreduction and diff constraints} we provide a geometric characterization of differential constraints   for systems of evolution PDEs of the form \refeqn{equation_problem}.  Hence, in Section \ref{Section Characteristic}, we discuss the problem of characteristics in  $\jinf(M,N)$ and in Section \ref{Building differential constraints} we apply previous results to the explicit construction of differential constraints (or reduction functions) for  evolution PDEs of the form \refeqn{equation_problem}. Finally, in Section \ref{section_examples}, we apply the results of the previous sections  to several  explicit examples.

\section{Preliminaries}\label{prel}

In this section we collect  some basic facts about (infinite) jet bundles in order to provide the necessary geometric tools for our aims.\\

Given the trivial fiber bundle $M \times N \rightarrow M$, where $M,N$ are  open sets of $\mathbb{R}^m$ and $\mathbb{R}^n$ respectively, we denote by $x^i$
the cartesian coordinate system of $M$  and  by $x^i, u^j$ the cartesian coordinate system of $M \times N$. The coordinates $x^i,u^j$ induce a global
coordinate system  $x^i,u^j,u_{\sigma}^j$ on  the $k$-order jet bundle $J^k(M,N)$,  where $\sigma=(\sigma_1,...,\sigma_m) \in \mathbb{N}_0^m, \vert \sigma\vert \le k$
is a multi-index denoting the number $\sigma_l$ of derivatives of $u^j$ with respect $x^l$. \\
It is well known that  $J^k(M,N)$ admits a natural structure of finite-dimensional smooth vector bundle on $M$ and,  considering the natural projections
$\pi_{k,h}:J^k(M,N) \rightarrow J^h(M,N)$, it is possible to define the inverse limit $\jinf(M,N)$ of the sequence
$$M \stackrel{\pi_0}{\leftarrow} M \times N=J^0(M,N) \stackrel{\pi_{1,0}}{\leftarrow} J^1(M,N) \stackrel{\pi_{2,1}}{\leftarrow} ...
\stackrel{\pi_{k,k-1}}{\leftarrow} J^k(M,N) \stackrel{\pi_{k+1,k}}{\leftarrow} ... $$
Unfortunately the space $\jinf(M,N)$ is not a finite-dimensional manifold, being the inverse limit of a sequence of  spaces of increasing dimension.
From a topological point of view $\jinf(M,N)$ is a Fr\'echet manifold modeled on $\mathbb{R}^{\infty}$ (see \cite{Saunders}) and any open set $U$ of $\jinf(M,N)$
contains a set of the form $U'=\pi^{-1}_k(V)$ for some $k \in \mathbb{N}$ and some open set $V \subset J^k(M,N)$. This means that, for any $\sigma$ with $|\sigma|>k$,
the coordinates $u^i_{\sigma}$ vary in the whole $\mathbb{R}$. \\
Furthermore,  the differential structure of $J^k(M,N)$ induces a  natural  differential structure in $\jinf(M,N)$
(see \cite{KrasVin2,KrasVin} for a complete description). \\
Hereafter we denote by $\mathcal{F}_k$ the algebra of real-valued smooth functions defined on $J^k(M,N)$   and we deduce  the differential structure of $\jinf(M,N)$ from the geometric smooth algebra $\mathcal{F}$ defined as the direct limit of the sequence
$$\cinf(M) \stackrel{\pi^*_0}{\rightarrow} \mathcal{F}_0 \stackrel{\pi^*_{1,0}}{\rightarrow} ... \stackrel{\pi^*_{k,k-1}}{\rightarrow} \mathcal{F}_k \stackrel{\pi^*_{k+1,k}}{\rightarrow} ...$$
Let $\mathcal{G} \subset \mathcal{F}$ be a finitely generated subalgebra of $\mathcal{F}$,  which means that there are a finite number of functions $g_1,...,g_l \in \mathcal{G}$ such that any  $g \in \mathcal{G}$ is of the form $g=G(g_1,...,g_l)$ for a unique smooth function $G$. It is possible to associate with $\mathcal{G}$ in a unique way a finite dimensional manifold $M_{\mathcal{G}}$ (see \cite{Jet}). For this reason in the following we identify the subalgebra $\mathcal{G}$ with the manifold $M_{\mathcal{G}}$ such that $\mathcal{G}=\cinf(M_{\mathcal{G}})$. The inclusion $i:\mathcal{G} \rightarrow \mathcal{F}$ induces a unique projection $\tilde{\pi}:\jinf(M,N) \rightarrow M_{\mathcal{G}}$ such that $\tilde{\pi}^*=i$.\\
The algebra  $\mathcal{F}$ is a graded algebra and a vector field $X$ on $\jinf(M,N)$ is a derivation  on the space  $\mathcal{F}$ which respects the order. \\
It is well known that the Cartan  distribution $\mathcal{C}$ on  $\jinf(M,N)$  generated by the vector fields
$$D_i=\partial_{x^i}+\sum_{k,\sigma} u^k_{\sigma+1_i} \partial_{u^k_{\sigma}}$$
defines an integrable connection on $\jinf(M,N)$. Hence, for any vector field $X$ on $\jinf(M,N)$,
we can write
$$X=X_v+X_h,$$
where $X_h \in \mathcal{C}$ and $X_v$ is a  vertical vector field i.e. $X_v(x^i)=0$ for any $i=1,...,m$.\\
A vector field $X$ on $\jinf(M,N)$ is  a symmetry of the Cartan distribution if $[X,\mathcal{C}] \subset \mathcal{C}$.
We remark that if $X=X_v+X_h$ is a symmetry of $\mathcal{C}$ then also $X_h$ and $X_v$ are symmetries of $\mathcal{C}$.
\begin{definition}\label{definition evolution vf}
A vertical vector field $X$ that is a symmetry of $\mathcal{C}$ is called \emph{evolution vector field}.
If $X$ is an evolution vector field there exists a unique smooth function $F:\jinf(M,N) \rightarrow \mathbb{R}^n$  such that
\begin{equation}\label{evolutionVF}
X=\sum_{i,\sigma} D^{\sigma}(f^i)\partial_{u^i_{\sigma}},
\end{equation}
where $F=(f^1,...,f^n)$ and $D^{\sigma}=(D_1)^{\sigma_1}...(D_m)^{\sigma_m}$. We call $F$ the \emph{generator of the evolution vector field} $X$ and we write $X=V_F$.\\
\end{definition}
If $V_F$ and $V_G$ are two evolution vector fields, then $[V_F,V_G]$ is also an evolution vector field and there exists an unique $H \in \mathcal{F}^n =\mathcal{F} \times ... \times \mathcal{F}$ such that $[V_F,V_G]=V_H$. Therefore the commutator between evolution vector fields induces a commutator in $\mathcal{F}^n$ and we define $[F,G]=H$  when  $[V_F,V_G]=V_H$.\\

\noindent We conclude this section recalling that  a subset $\mathcal{E}$  of $\jinf(M,N)$ is
a submanifold of $\jinf(M,N)$ if for any $p \in \mathcal{E}$ there exists a neighborhood $U_p$ of $p$ such that $\pi_h(\mathcal{E} \cap U_p)$ is a
submanifold of $J^h(M,N)$ for $h>H_p$.\\
If, for any $p \in \mathcal{E}$, all the submanifolds $\pi_h(\mathcal{E} \cap U_p)$ with $h>H_p$   have the
 same dimension $L$, we say that $\mathcal{E}$ is an $L$-dimensional submanifold of $\jinf(M,N)$.
In particular, given an $L$-dimensional  manifold $B$ and a smooth immersion $K:B \rightarrow \jinf(M,N)$, for any point $y \in B$ there exists a neighborhood $V$ of $p$ such that $K(V)$ is a finite dimensional submanifold of $\jinf(M,N)$.
\begin{definition}\label{definition canonical_submanifold}
A submanifold $\mathcal{E}$ of $\jinf(M,N)$ such that $\mathcal{C} \subset T\mathcal{E}$ is said \emph{canonical submanifold}.  Any canonical submanifold $\mathcal{E}$ can  be locally described as the set of zeros of a finite number of smooth independent functions $f_1,...,f_L$ and of all their differential consequences $D^{\sigma}(f_i)$.\\
\end{definition}
\bigskip

\section{ Differential constraints and PDEs reduction }\label{PDEreduction and diff constraints}

In this section we propose a geometric reformulation of differential constraints method for a  family of  evolution PDEs. In particular we introduce  the notion of reduction function and we discuss its relation with differential constraints seen as finite dimensional submanifolds of $\jinf(M,N)$.

\subsection{Differential constraints: from the reduction function to the submanifold}

Let us consider a system of evolution PDEs  of the form
\begin{equation}\label{equation_evolution}
\partial_t(u^k)=\sum_{i=1}^s c^i(t) F^k_i(x,u,u_{\sigma}),
\end{equation}
where $F_i \in \mathcal{F}^n$ and $k=1, \ldots, n$.
\begin{definition}
Given a system of  evolution PDEs of the form \refeqn{equation_evolution} and
 an $L$-dimensional  manifold $B$, let
$$K:M \times B \rightarrow N$$
be a smooth function. We say that $K$ is
a \emph{reduction function} for  \refeqn{equation_evolution} if there are smooth functions $f^j$ such that
$$U(x,t)=K(x,b^1(t),...,b^L(t))$$
is a solution to the system \refeqn{equation_evolution} for any  $c^1(t),...,c^s(t)$ if and only if $b^j(t)$ are solutions to the system of ODEs
$$\partial_t(b^j)=f^j(t,b,c^1(t),...,c^s(t)).$$
\end{definition}
In this framework it is natural to associate  with  $K$
a function $R^K:M \times B \rightarrow \jinf(M,N)$ that is the lift of  $K$ to $\jinf(M,N)$. In particular, if $\pi_0:\jinf(M,N) \rightarrow J^0(M,N)=M \times N$ denotes the natural projection of $\jinf(M,N)$ onto $J^0(M,N)$, the function  $R^K $ satisfies
\begin{eqnarray*}
\pi_{0} \circ R^K(x,b)=(x,K(x,b))\\
R^K_*(\partial_{x^i})=D_i,
\end{eqnarray*}
 and in coordinates we have
$$(u^i_{\sigma}\circ R^K)(x,b)=\partial^{\sigma}_x(K^i(x,b)).$$
If $R^K$ is an immersion,  then $\mathcal{K}=R^K(M \times B)$ is (possibly restricting  $B$)  a finite dimensional submanifold of $\jinf(M,N)$. Furthermore the following theorem holds.

\begin{theorem}\label{theorem_reduction}
Let $K:M \times B \rightarrow N$ be a smooth function and $\mathcal{K}=R^K(M \times B)$.
Then $K$ is a reduction function for the system \refeqn{equation_evolution} if and only if
$V_{F_i} \in T\mathcal{K}$, $\forall i=1, \ldots, s$.
\end{theorem}
\begin{proof}
If $K$ is a reduction function for the system \refeqn{equation_evolution}, we have
$$\sum_j f^j(x,b,c)\partial_{b^j}(K^l)(x,b)=\sum_{i=1}^s c^i(t)F^l_i(x,K(x,b),\partial^{\sigma}_x(K)(x,b)).$$
Choosing $c^i=\delta_{1,i}$ and applying   $\partial^{\sigma}_x$ to both sides of the previous equation
we get
$$\sum_j f^j \partial_{b^j}(u^l_{\sigma} \circ (R^K))=D^{\sigma}(F^l_1) \circ R^K.$$
Since $R^K_*(\partial_{b^i})=\partial_{b^i}(R^K) \in T\mathcal{K}$ and  $V_{F_1}(u^l_{\sigma})|_{\mathcal{K}}=D^{\sigma}(F^l_1)|_{\mathcal{K}}$,  we have $V_{F_1} \in T\mathcal{K}$. Choosing $c^i=\delta_{p,i}$ we obtain $V_{F_p}\in T\mathcal{K}$.\\
Conversely suppose that $V_{F_i} \in T\mathcal{K}$. Since $V_{F_i}$ are vertical and the vertical vector fields of $T\mathcal{K}$ are generated by $\partial_{b^i}(R^K)$ there exist suitable functions $g_i^j:M \times B \rightarrow \mathbb{R}$ such that
$$V_{F_j}=\sum_i g^i_j\partial_{b^i}(R^K).$$
It is easy to show that the functions $g^j_i$ do not depend on $x \in M$ being   $R^K_*(\partial_{x^i})=D_i$ so the thesis follows choosing
$$f^j(b,c^1,...,c^N)=\sum_i c^i g^j_i(b).$$
\hfill\end{proof}

\noindent  Since the submanifold $\mathcal{K}=R^K(M \times B)$ is  a finite dimensional canonical submanifold, $\mathcal{K}$ can  be locally described as the set of zeros of a finite number of smooth independent functions $f_1,...,f_L$ and of all their differential consequences $D^{\sigma}(f_i)$. Therefore
 a necessary and sufficient condition for $V_F \in T\mathcal{K}$ is  $V_F(D^{\sigma}(f_i))|_{\mathcal{K}}=0$ but, since $D_i$ and $V_F$ commute
and $D_i \in T\mathcal{K}$, it is sufficient to check that $V_F(f_i)|_{\mathcal{K}}=0$.\\

\begin{remark}
In the proof of Theorem \ref{theorem_reduction} the hypothesis that $\mathcal{K}$ is a submanifold of $\jinf(M,N)$ is not necessary.
Indeed  we prove  that $V_{F_i} \in \operatorname{Image}(TR^K)$ even if $R^K$ is not an immersion (and so $\mathcal{K}$ is not a submanifold).
Even so, for  the sake of simplicity,  in the following we always consider submanifolds $\mathcal{K}$ of $\jinf(M,N)$. \\
In particular, if $K$ is a real analytic function with respect the $x^i$ variables, a necessary and sufficient condition for $\mathcal{K}$ to be a submanifold
is that $\partial_{b^i}(K(\cdot,b))$  are linearly independent as functions from $M$ into $N$.
In the smooth case it can happen that $\partial_{b^i}(K(\cdot,b))$ are linearly independent but $R^K$ is not an immersion.
However, this situations can be considered as exceptional:  indeed the set of $K$ such that $R^K$ is an immersion is an open everywhere dense subset of $\cinf(M,N)$ with respect the Whitney topology (see \cite{GolGui}).
\end{remark}

\begin{remark}\label{remark_one_equation}
An interesting consequence of Theorem \ref{theorem_reduction} is that, if $s=1$,  any solution $U(x,t):M \times \mathbb{R} \rightarrow N$
 to the system \refeqn{equation_evolution} for $c_1=1$ is a reduction function. Indeed in this case we have that
$$\partial_b(U(x,b))=F_1(x,U,U_{\sigma}),$$
so $\partial^{\sigma'}_x(\partial_b(U(x,b)))=D_{\sigma'}(F_1)(x,U,U_{\sigma})$. Hence $R^U_*(\partial_b)=V_{F_1}$ and  equation \eqref{equation_evolution} with $s=1$  becomes and ODE for $b$  of the form
$$\partial_t(b)(t)=c^1(t).$$
\end{remark}

\subsection{Differential constraints: from the submanifold to the  reduction function}

In this section we discuss the problem of computing  the reduction function  $K$ starting from the knowledge of  a suitable canonical submanifold $\mathcal{K}$. The resulting algorithm will be used in the examples of Section \ref{section_examples}.

\begin{definition}
A finite dimensional canonical submanifold $\mathcal{K}$ of $\jinf(M,N)$ is a \emph{differential constraint} for  equation \refeqn{equation_evolution} if $V_{F_i} \in T\mathcal{K}$
\end{definition}
In order to prove that with  any differential constraint $\mathcal{K}$  for \refeqn{equation_evolution} it is possible to associate a reduction function $K$, we need to recall the following technical result.
\begin{theorem}\label{theorem_submanifold}
Let $\mathcal{H}$ be an $m$-dimensional canonical submanifold of $\jinf(M,N)$ (i.e. $T\mathcal{H}=\mathcal{C}$). Denoting by  $\pi$
 the canonical projection $\pi: \jinf(M,N) \to M$, if   $\pi(\mathcal{H})=V \subset M$, then there exists a unique smooth function $U:V \rightarrow N$ such that $R^U(V)=\mathcal{H}$.
\end{theorem}
\begin{proof}
 A proof of this result can be found in \cite{KrasVin}, Chapter 4, Proposition 2.3.
\hfill\end{proof}

\begin{theorem}\label{theorem_connection}
 Let $\mathcal{K}$ be a finite dimensional canonical submanifold of $\jinf(M,N)$  which is a differential constraint for  equation \refeqn{equation_evolution}.
Then, for any point $p \in \mathcal{K}$, there exist a neighborhood $U \subset \jinf(M,N)$ of $p$ and a function $K:V \times B \rightarrow N$, where $V \subset \pi(U)$, such that $R^K(V \times B)=\mathcal{K} \cap U \cap \pi^{-1}(V)$ and  $K$ is a reduction function for the system \refeqn{equation_evolution}.
\end{theorem}
\begin{proof}
The manifold $\mathcal{K}$ with respect the projection $\pi$ is a finite dimensional fibred manifold with base $M$. Moreover, since $D_i \in \mathcal{K}$,
the Cartan distribution $\mathcal{C}$ is a finite dimensional flat connection of  $(\mathcal{K}, \pi, M)$ and,
 for any $p_0 \in \mathcal{K}$, there exist a neighborhood $B \subset \pi^{-1}(x_0)$ of $p_0$ (where $x_0 =\pi(p_0)$) and a local trivialization
$R:V \times B \rightarrow \mathcal{K} \subset \jinf(M,N)$ of  $\mathcal{C}$.
Therefore Theorem \ref{theorem_submanifold} ensures that there exists a smooth function $K:V \times B \rightarrow N$ such that $R^K=R$ and
Theorem \ref{theorem_reduction} guarantees that $K$ is also a reduction function for the system \refeqn{equation_evolution}.
\hfill \end{proof}

\begin{remark}
There are two obstructions for   a global version of  Theorem \ref{theorem_connection}.
The first one is that, if $M$ is not simply connected,   $\mathcal{C}$ may admit only a local trivialization and not a global one and
the second  is that, if $\mathcal{C}$ is a non-linear connection on $\mathcal{K}$, it may not admit a global trivialization, since
 non-linear ODEs can blow-up. Obviously, if $M$ is simply connected and $\mathcal{C}|_{\mathcal{K}}$ has at most linear grow for some coordinate system on $\mathcal{K}$, Theorem \ref{theorem_connection} admits a global version.
\end{remark}
Given an $L$-dimensional canonical submanifold $\mathcal{K}$ which is a differential constraint for \eqref{equation_evolution},
Theorem \ref{theorem_connection} provides an explicit construction procedure for  the function $K$  and for the system of   ODEs for the parameters $b^i$.\\
Indeed let $(x^i,y^j)$ be an adapted coordinate system for the fibred manifold $(\mathcal{K},\pi,M)$, which means that $x^i$ is the standard coordinate
system on $M$ and $y^j$ is an adapted coordinate system for the fiber $\pi^{-1}(x)$, with $j=1,...,L-m$.
The coordinates $y^i$ can be chosen among the functions  $u^k,u^l_{\sigma}$ and, in general,
it is possible to find smooth functions $f^k(x,y),f^l_{\sigma}(x,y)$ such that
$u^k=f^k(x,y),u^l_{\sigma}=f^l_{\sigma}(x,y)$.\\
In the coordinate system $(x^i,y^j)$ the vector fields $D_i$ have the form
$$D_i=\partial_{x^i}+\sum_j \Psi^j_i(x,y) \partial_{y^j} \qquad (i=1, \ldots m)$$
and, $\forall x_0 \in M$ and $y_{x_0} \in \pi^{-1}({x_0})$, the solution  $y^j=\tilde{K}^j(x,y_{x_0})$ to the system
\begin{eqnarray*}
\partial_{x^i}(\tilde{K}^j(x,y_{x_0}))&=&\Psi^j_i(x,y_{x_0})\\
\tilde{K}^j(x_0,y_{x_0})&=&y^j_{x_0}
\end{eqnarray*}
provides the local trivialization of the flat connection $\mathcal{C}=\spann\{D_i\}$.
The explicit expression of the function $K$ can be obtained by
rewriting  $u^l$ as functions of $(x^i,y^j)$ leading to
$$K^l(x,y_{x_0})=f^l(x,\tilde{K}(x,y_{x_0})).$$
Moreover  the system of ODEs for the parameters $y^j_{x_0}$ can be obtained expressing $V_{F_i}$ in the coordinates $(x^i,y^j)$
$$V_{F_i}=\sum_j V_{F_i}(y^j) \partial_{y^j}= \sum_j \Phi^j_i(x,y)\partial_{y^j},$$
so that
$$\frac{dy^j_{x_0}}{dt}=\sum_i c^i(t)\Phi^j_i(x_0,y_{x_0}(t)).$$

\subsection{A necessary condition for existence of differential constraints}

The main goal of the general theory of  differential constraints is to find  a reduction function for a system of the form \eqref{equation_evolution}  for $s=1$. As proven in \cite{Olver1} (see also Remark \ref{remark_one_equation}) this problem admits an infinite number of solutions. On the other hand, the problem of existence of reduction functions (or  differential constraints) for $s>1$ is  completely different: actually,  in the general case, there are no reduction functions at all. \\
In this section we address the problem of existence of a reduction function for a system of the form \eqref{equation_evolution}  starting from the following remark.
\begin{remark}
If a system of evolution PDEs of the form \eqref{equation_evolution}  admits a differential constraint $\mathcal{K}$, the
set
$$\mathcal{S}=\spann\{V_{F_1},...,V_{F_s}\},$$
is a finite dimensional module on $\mathcal{K}$.
\end{remark}
The following Proposition provides a useful characterization for the vector fields $V_{F_i}$.

\begin{proposition}\label{proposition_lie_algebra}
Let $V_{F_1},...,V_{F_s}$ be evolution vector fields in $\jinf(M,N)$ such that $\mathcal{S}$ is an $s$-dimensional (formally) integrable distribution on a
submanifold $\mathcal{K}$ of $\jinf(M,N)$. If
$$[V_{F_i},V_{F_j}]=\sum_h \lambda_{i,j}^hV_{F_h}$$
then $D_l(\lambda_{i,j}^h)=0$ on $\mathcal{K}$.
\end{proposition}
\begin{proof}
The proof is given  for the case $N=M=\mathbb{R}$ and $\mathcal{H}=\jinf(M,N)$;  the general case is a simple generalization of this one.\\
Since $\mathcal{S}$ is $s$-dimensional, for any point $p \in \jinf(M,N)$ there exist a neighborhood $U$ of $p$ and an integer
$h \in \mathbb{N}_0$ such that the matrix $A=(D_x^{h+j-1}(F_i))|_{i,j=1,...,s}$ is non-singular. Moreover, since
the commutator of two evolution vector fields is an evolution vector field, there exist some $F_{i,j} \in \mathcal{F}$ such that $[V_{F_i},V_{F_j}]=V_{F_{i,j}}$ and, by the definition of evolution vector field,  we have
\begin{equation}\label{CN1}
D^r_x(F_{i,j})=\sum_h \lambda^h_{i,j} D^r_x(F_h).
\end{equation}
Deriving with respect to $x$ the previous relations we obtain
\begin{equation}
\label{CN2}D^{r+1}_x(F_{i,j})=\sum_h D_x(\lambda^h_{i,j}) D^r_x(F_h)+\sum_h \lambda^h_{i,j} D^{r+1}_x(F_h)
\end{equation}
and combining \eqref{CN1} and \eqref{CN2} we find
$$\sum_h D_x(\lambda^h_{i,j}) D^r_x(F_h)=0.$$
Since the matrix $A$ is non-singular we get  $D_x(\lambda^h_{i,j})=0$.
\hfill\end{proof}

\medskip

\noindent In Section \ref{Building differential constraints}
 we will consider two particular cases for the functions $F_i \in \mathcal{F}^n$.\\
In the first case $V_{F_i}$ form a finite dimensional module of constant dimension on all $\jinf(M,N)$.
In this case Proposition \ref{proposition_lie_algebra} ensures that $V_{F_i}$  form a Lie algebra, since the only functions $f$ in $\jinf(M,N)$ such that $D_i(f)=0$ are the constants.\\
In the second case we suppose that $V_{F_i}$ form a finite dimensional module on a real analytic finite dimensional submanifold $\mathcal{H}$ of $\jinf(M,N)$.

\section{Characteristic vector fields in $\jinf(M,N)$}\label{Section Characteristic}

In this section we define the notion of generalized characteristic flow for an evolution vector field and we discuss the connection with the usual characteristic flow for scalar first order evolution PDEs. These results will play a central role in the explicit construction of differential constraints in Section \ref{Building differential constraints}.

\subsection{Characteristics of scalar first order evolution PDEs}\label{subsection_scalar}

 It is well known that, if $N=\mathbb{R}$ and $F \in \mathcal{F}\setminus\mathcal{F}_0$, the evolution vector field $V_F$ is not the prolongation  of a vector field on $J^0(M,N)$ and  does not admit flow in $\jinf(M,N)$, which is  why the equation
\begin{equation}\label{pippo}
\partial_t(u)=F(x,u,u_{\sigma})
\end{equation}
may not admit  solutions even for smooth initial data, or may admit infinite solutions for any smooth initial data. For this reason
the problem of finding solutions to  evolution PDEs is usually solved only in specific situations (for example the linear or semilinear cases) where it is possible to use the powerful techniques of analysis.\\
Anyway,  a classical  geometric approach to scalar first order evolution PDEs (see, e.g.,  \cite{Hilbert}) shows that something can be done in order to solve equation \eqref{pippo} even when $V_F$ does not admit flow in $\jinf(M,N)$.
Indeed given a first order scalar autonomous PDE
\begin{equation}\label{equation_first_order}
\partial_t(u)=F(x^j,u,u_{i})
\end{equation}
it is  possible to solve \refeqn{equation_first_order} considering the following  system of ODEs  on $J^1(M,N)$
\begin{eqnarray*}
\frac{dx^i}{da}&=&-\partial_{u_{i}}(F)(x^j,u,u_{k})\\
\frac{du}{da}&=&F(x^j,u,u_{k})-\sum_h u_{h}\partial_{u_{h}}(F)(x^j,u,u_{k})\\
\frac{du_{i}}{da}&=&\partial_{i}(F)(x^j,u,u_{k})+u_{i}\partial_{u}(F)(x^j,u,u_{k}).
\end{eqnarray*}
If $\Phi_a$ is flow of the vector field on $J^1(M,N)$ corresponding to the previous system and we define $\phi_a^i=(\Phi_a^*(x^i))$ and $\eta_a=\Phi_a^*(u)$, the solution $U(x,t)$ to the
PDE \refeqn{equation_first_order} with initial  data $U(x,0)=f(x)$ is given by
$$U(x,t)=\eta_t(\bar{\phi}_t^{-1}(x),f(\bar{\phi}_t^{-1}(x)),\partial_{x^i}(f)(\bar{\phi}_t^{-1}(x)))$$
where $\bar{\phi}_a(x)=\phi_a(x,f(x),\partial_{j}(f)(x))$.\\
Moreover it is possible to uniquely extend the flow $\Phi_a$ to $J^k(M,N)$
as the solution to the following system of ODEs
$$\frac{du_{\sigma}}{da}=D^{\sigma}(F)(x,u,u_{\sigma})-\sum_i u_{\sigma+1_i} \partial_{u_{i}}(F)(x,u,u_{\sigma}).$$
Defining $\psi_{\sigma,a}=\Phi_a^*(u_{\sigma})$ we have
$$\partial^{\sigma}(U)(x,t)=\psi_{\sigma,t}(\bar{\phi}_t^{-1}(x),f(\bar{\phi}_t^{-1}(x)),\partial^{\sigma}(f)(\bar{\phi}_t^{-1}(x))),$$
and the vector field corresponding to the flow $\Phi_a$ on $\jinf(M,N)$ is given by
$$\bar{V}_F:=\partial_a(\Phi_a)|_{a=0}=V_F-\sum_i\partial_{u_{i}}(F)D_i.$$
We call $\Phi_a$ the \emph{characteristic flow} of $F$ and $\bar{V}_{F}$ its \emph{characteristic vector field}.

\subsection{Characteristics in the general setting}
In this section  we propose an extension of the notion
of characteristic vector field and characteristic flow to multidimensional and higher order case. This extension is based on the
geometric analysis of $\jinf(M,N)$ presented in \cite{KrasVin2}.
We start by recalling the definition  of one-parameter group of local diffeomorphisms on $\jinf(M,N)$ which reduces to the classical one   in the finite dimensional setting.

\begin{definition}\label{definition_one-parameter flow}
A map $\Phi_a:U_a \rightarrow \jinf(M,N)$ is a one-parameter group of local diffeomorphisms if
 $\Phi_a$ are smooth maps,
$U_a$ are open sets $\forall a$
(with $U_0=\jinf(M,N)$) and $\forall p \in U_{a+b} \subset U_b \cap \Phi^{-1}_{b}(U_a)$ (with  $ab\ge0$) we have
 $\Phi_a \circ \Phi_b(p)=\Phi_{a+b}(p)$.\\
The one-parameter group $\Phi_a$ of local diffeomorphisms is the flow of the vector field $X$ if
$$\partial_a(\Phi_a^*(f)(p))|_{a=0}=X(f)(p)$$
for any $f \in \mathcal{F}$.
\end{definition}

\begin{definition}\label{definition_charatcteristic}
Given an evolution vector field $V_F$, we say that $V_F$ (or its generator  $F$) \emph{admits characteristics} if there exist suitable  smooth functions $h^1,...,h^n \in \mathcal{F}$ such that the vector field
$$\tilde{V}_F=V_F-\sum_i h^iD_i,$$
admits flow on $\jinf(M,N)$.
\end{definition}
If we restrict to the scalar case ($N=\mathbb{R}$), which is discussed in Subsection \ref{subsection_scalar},  the following Theorem provides a complete characterization of evolution vector fields admitting characteristics.

\begin{theorem}\label{theorem_first_order}
An evolution vector field $V_F$ on $\jinf(M,\mathbb{R})$ with generator $F$ admits characteristic if an only if $F \in \mathcal{F}_1$.
\end{theorem}
\begin{proof}
The proof that any $F \in \mathcal{F}_1$  admits characteristic flow is given in  Subsection \ref{subsection_scalar}. The proof of the converse can be found in \cite{KrasVin2}.
\hfill\end{proof}

\begin{remark}
Theorem \ref{theorem_first_order} does not hold if, instead of requiring that $\tilde{V}_F$ admits flow on the whole $\jinf(M,\mathbb{R})$,
we  restrict  to a submanifold of $\jinf(M,\mathbb{R})$.
For example if we consider  $M=\mathbb{R}^2$ with coordinates $(x,y)$ and $F=u_{xy}$,  Theorem \ref{theorem_first_order}
ensures  that $F=u_{xy}$ does not admit characteristics on $\jinf(\mathbb{R}^2,\mathbb{R})$ but,  considering the submanifold
$\mathcal{E} \subset \jinf(\mathbb{R}^2,\mathbb{R})$ generated by the equation $u_{yy}=0$,  it is easy to prove that $V_F \in T\mathcal{E}$ and that $V_F$ admits characteristics  on $\mathcal{E}$.
\end{remark}

If we do not restrict to the scalar case  the situation becomes more complex and, to the best of our knowledge,  a complete theory  of characteristics in $\jinf(M,N)$ for  $N  \not = \mathbb{R}$ has not been developed.\\
Indeed in this case we can find
$F \not \in \mathcal{F}_1^n$ such that  $V_F$ admits characteristics.
For example if we consider $M=\mathbb{R}$ and $ N=\mathbb{R}^2$ (with coordinates $x$ and $(u,v)$ respectively) and $F=(v_{xx},0) \in  \mathcal{F}_2$, the flow of the vector field $V_F$ is given by the following transformation
\begin{eqnarray*}
x_a&=&x\\
u_a&=&u+av_{xx}\\
u_{x,a}&=&u_x+av_{3}\\
u_{xx,a}&=&u_{xx}+av_4\\
&...&\\
v_a&=&v\\
v_{x,a}&=&v_x\\
&...&
\end{eqnarray*}
In this paper, in order to deal with the general case, we propose a stronger definition of characteristics that, although imitating  in some respects the scalar case,  is weak enough to include many cases of interest. \\
Given an open subset $U\subset \jinf(M,N)$ we denote by
$$\mathcal{F}|_{U}=\bigcup_k \mathcal{F}_k|_{U}$$
the set of smooth functions defined on  $U$, that is the union of the sets of smooth functions defined on $\pi_k(U) \subset J^k(M,N)$.\\
Given a subalgebra $\mathcal{G}_0 \subset \mathcal{F}|_U$, we denote by $\mathcal{G}_k$ the algebra generated by smooth  composition of functions of the form $D^{\sigma}(f)$, where $f \in \mathcal{G}_0$ and $\sigma$ is a multi-index with $|\sigma|\leq k$.

\begin{definition}\label{definition_algebra_generate}
A subalgebra $\mathcal{G}_0 \subset \mathcal{F}|_U$  generates $\mathcal{F}|_U$ if $x^i \in \mathcal{G}_0$ and
$$\mathcal{F}|_U=\bigcup_k \mathcal{G}_k.$$
\end{definition}

\begin{definition}\label{definition_strong}
An evolution vector field $V_F$ with generator $F$ admits strong characteristics  if there exists an open set $U \subset \jinf(M,N)$, a finitely generated subalgebra $\mathcal{G}_0$ of $\mathcal{F}|_{U}$ generating $\mathcal{F}|_U$  and  $g^1,...,g^n \in \mathcal{F}$ such that
the vector field $\bar{V}_F=V_F-\sum_i g^i D_{i}$ satisfies
$$\bar{V}_F(\mathcal{G}_0)\subset \mathcal{G}_0.$$
\end{definition}
In the scalar case an evolution vector field admits characteristics if and only if  admits strong characteristics:  indeed in this case $\bar{V}_{F}(\mathcal{F}_1) \subset \mathcal{F}_1$. Moreover, if we consider the evolution vector field $V_F$ of the previous example (with generator $F=(v_{xx},0)$), it is easy to check that $V_F$ admits strong characteristics.
 In fact the subalgebra  $\mathcal{G}_0$ generated by $x,u,v,v_x,v_{xx}$ is such that $V_F(\mathcal{G}_0) \subset \mathcal{G}_0$. Actually we do not know any example of evolution vector field admitting characteristics which are not strong characteristics.

\begin{remark}
In Definition \ref{definition_strong} we consider  a general subalgebra $\mathcal{G}_0$ generating $\mathcal{F}$ instead of restricting to the case  $\mathcal{G}_0=\mathcal{F}_k$ for some $k \in \mathbb{N}$. This is a crucial point because, in the vector case $N=\mathbb{R}^n$ (with $n>1$),  condition $\bar{V}_F(\mathcal{F}_k) \subset \mathcal{F}_k $ implies   $\bar{V}_F(\mathcal{F}_0) \subset \mathcal{F}_0$ (see \cite{KrasVin2}) and the vector fields $\bar{V}_F$ satisfying $\bar{V}_F(\mathcal{F}_k) \subset \mathcal{F}_k $ turn out to  be tangent to the prolongations of  infinitesimal transformations in $J^0(M,N)$.\\
A well-known consequence of this fact is that, in the vector case, the only infinitesimal symmetries of a PDE which can be defined using finite jet spaces $J^k(M,N)$ are Lie-point symmetries. On the other hand,  if we allow $\mathcal{G}_0$ to be a general subalgebra generating $\mathcal{F}$, we obtain a larger and non-trivial class of evolution vector fields admitting strong characteristics.
\end{remark}

\begin{theorem}\label{theorem_strong_characteristic}
With the notations of Definition \ref{definition_strong}, if an evolution vector field admits strong characteristics then it admits characteristics, and $\bar{V}_F$ is its characteristic vector field.
\end{theorem}
\begin{proof}
The vector field $\bar{V}_F$ admits flow on the space of functions $\mathcal{G}_0$ since $\mathcal{G}_0$ is finite dimensional. In order to show that $\bar{V}_F$ admits flow on all $\mathcal{F}|_U$ and so (since $U$ depends  on a generic point) on $\mathcal{F}$ we prove by induction that $\bar{V}_F(\mathcal{G}_k) \subset \mathcal{G}_k$. \\
By hypothesis $\bar{V}_F(\mathcal{G}_0)\subset \mathcal{G}_0$. Suppose that $\bar{V}_F(\mathcal{G}_{k-1}) \subset \mathcal{G}_{k-1}$. Since $\bar{V}_F$ is a symmetry of the Cartan distribution,  there exist some functions $h_i^j \in \mathcal{F}$ such that
$$[\bar{V}_F,D_i]=\sum_j h^j_i D_j$$
where $h^i_j \in \mathcal{G}_1$, being $\bar{V}_F(\mathcal{G}_0)\subset \mathcal{G}_0$ and $x^i \in \mathcal{G}_0$. \\
We recall that $\mathcal{G}_k$ is generated by functions of the form $D_i(g)$ with $g \in \mathcal{G}_{k-1}$. So
$$\bar{V}_F(D_i(g))=D_i(\bar{V}_F(g))+\sum_j h^j_i D_j(g) \in \mathcal{G}_{k}$$
since $\bar{V}_F(g) \in \mathcal{G}_{k-1}$ and $h^i_j \in \mathcal{G}_1$.
Hence $\bar{V}_F$ admits flow on $\mathcal{G}_k$ and the flow on $\mathcal{G}_k$ is compatible with the flow on $\mathcal{G}_{k-1}$ being $\mathcal{G}_{k-1} \subset \mathcal{G}_k$.\\
The problem of the previous construction is that in general the domain $U_k$ of the flow in $\mathcal{G}_k$ depends on  $k$. This means that, if we denote with $P_{h,k}$ the natural projection of $\mathcal{G}_h$ on $\mathcal{G}_k$ with $h>k$, it might happen that $P_{h,k}^{-1}(U_k) \not = U_h$. But this is not actually the case. Indeed since $\mathcal{G}_0$ generates $\mathcal{F}|_U$, then $\mathcal{F}_0|_U \subset \mathcal{G}_0$ and so $\mathcal{F}_k|_U \subset \mathcal{G}_k$. In particular $u^i_{\sigma} \in \mathcal{G}_k$ if $|\sigma|\leq k$. But by Remark \ref{remark_flow1} and Corollary \ref{corollary_flow1} (see Appendix) $\Phi_a(u^i_{\sigma})$ is polynomial in $u^j_{\sigma'}$ for $|\sigma'|$ sufficiently large. This means that $u^j_{\sigma'}$ can vary in all $\mathbb{R}$ and so the domain of definition of $\Phi_a$ in $U \subset \jinf(M,N)$ is not empty and  is of the form $U'=\pi_{\infty, k}^{-1}(U_k)$ for $k$ sufficiently large. Since $U'$ is an open subset of $\jinf(M,N)$  this concludes the proof.
\hfill\end{proof}

\begin{definition}
Let $y^1,...,y^l,... \in \mathcal{F}|_U$  be a sequence of functions defined in an open set  $U$. We say that $Y=\{y^i\}|_{i \in \mathbb{N}}$
is a local adapted coordinate system with respect to a subalgebra $\mathcal{G}_0$  generating $\mathcal{F}|_U$, if there exists a sequence $k_1,...,k_l,... \in \mathbb{N}$, with $k_i < k_{i+1}$, such that $y^1,...,y^{k_i}$ is a coordinate system for $\mathcal{G}_i$.
\end{definition}

\begin{remark}
The flow of a vector field with strong characteristics solves a triangular infinite dimensional system of ODEs. Indeed if we consider an  adapted coordinate system with respect to a subalgebra $\mathcal{G}_0$
we have $\bar{V}_F(y^i)=f(y^1,..,y^{k_1})$ for $i=1,...,k_1$, $\bar{V}_F(y^i)=f(y^1,..,y^{k_2})$ for $i=k_1+1,...,k_2$  and so on. So we can start by solving the system for $i=1,...,k_1$ and then solve the system for $i=k_1+1,...,k_2$, since the system  is of triangular type.
\end{remark}

The main trouble when working with a family of evolution vector fields admitting characteristic flows is that the sum or the Lie brackets of two of them usually do not admit characteristic flow. In order to overcome this problem we give the following Definition.

\begin{definition}\label{definition_common_filtration}
A set of evolution vector fields $V_{F_1},...,V_{F_s}$ with  strong characteristics admits a \emph{common filtration} if  $\forall p \in \jinf(M,N)$ there exist a neighborhood $U$ of $p$ and a subalgebra $\mathcal{G}_0 \subset \mathcal{F}|_U$ such that $\mathcal{G}_0$ is the subalgebra required  in Definition \ref{definition_strong} for $\bar{V}_{F_1},...,\bar{V}_{F_s}$.
\end{definition}
If $F_1,...,F_s$, correspond to evolution vector fields with  strong characteristics admitting a common filtration, then also $c F_i + d F_j$ (where $c,d \in \mathbb{R}$) and $[F_i,F_j]$ correspond to  vector fields with  strong characteristics. Furthermore $c F_i+d F_j$ and $[F_i,F_j]$ admit the same common filtration of $F_1,...,F_k$.

\section{Building differential constraints}\label{Building differential constraints}
In this section we consider a system of  PDEs of the form \eqref{equation_evolution} such that  some of the evolution vector fields  $V_{F_i}$ admit strong characteristics and a common filtration. In this setting we show how it is possible to construct a differential constraint for the system \eqref{equation_evolution} starting from the knowledge of a suitable submanifold of $\jinf(M,N)$.
The construction, which  is completely explicit,  take the cue from  of the moving frame method (see \cite{Olver4}).

\begin{definition}\label{definition transersality}
Let  $\mathcal{H}\subset \jinf(M,N) $ be a submanifold and $U$ be an open neighborhood of $p \in \mathcal{H}$. Given a sequence of independent functions $f^i \in \mathcal{F}|_U$ ($i \in \mathbb{N}$) such that $\mathcal{H} \cap U$ is the annihilator of $f^i$, we say that a distributions  $\Delta= \spann\{V_{G_1}, \ldots, V_{G_h}\}$ is transversal to $\mathcal{H}$ in $U$ if there exist $r_1, \ldots , r_h$ such that
the matrix $(\bar{V}_{G_i}(f^{r_j}))|_{i, j=1,...,h}$ has maximal rank in $U$.
In the following  the sequence $f^i$ will be chosen so  that $r_j=j$ and $f^i$ is a local coordinate system adapted with respect to the filtration $\mathcal{G}_k$ for $k$ sufficiently large.
\end{definition}

\begin{lemma}\label{lemma_main}
Let $G_1,...,G_h$ be a subalgebra of $\mathcal{F}^n$  admitting strong characteristics and a common filtration. Let $\Phi^i_{a^i}$ be the characteristic flow of $G_i$ and $\mathcal{H}$ be a canonical finite dimensional submanifold of $\jinf(M,N)$ such that the distribution $T\mathcal{H} \oplus \spann\{V_{G_1},...,V_{G_h}\}$ has constant rank and the distribution
$\Delta= \spann\{V_{G_1}, \ldots, V_{G_h}\}$ is transversal to $\mathcal{H}$.   Then there exists a suitable neighborhood of the origin $\mathcal{V} \subset \mathbb{R}^{h}$  such that
\begin{equation}\label{manifold kappa}
\mathcal{K}=\bigcup_{(a^1,...,a^{h}) \in \mathcal{V}} \Phi^h_{a^h}(...(\Phi^{1}_{a^{1}}(\mathcal{H}))...)
\end{equation}
is  a finite dimensional submanifold of $\jinf(M,N)$.
\end{lemma}
\begin{proof}
In the following, for the sake of clarity,  we write
$$\bold{\Phi}_{\alpha}^*(f)=\Phi^{1*}_{a^1}(...\Phi^{h*}_{a^h}(f)...),$$
where $\alpha=(a^1,...,a^h) \in \mathbb{R}^h$.
Given a sequence of independent functions $f^i $ ($i \in \mathbb{N}$) such that $\mathcal{H} $ is the annihilator of $f^i$,
for any point $p \in \mathcal{H}$ there exists a neighborhood $U$ such that the matrix $(\bar{V}_{G_i}(f^j))|_{i,j=1,...,h}$  has maximal rank in $U$. Therefore, considering
the submanifold $\tilde{\mathcal{H}}$ defined as the annihilator of the functions $f^i \in \mathcal{F}|_U$ ($i=1, \ldots, h$),
the equations
\begin{equation}\label{manifold tildeH}
\bold{\Phi}^*_{\alpha}(f^i)=0 \qquad i=1, \ldots h
\end{equation}
can be solved with respect to $\alpha$. This means that, possibly restricting the open set $U$, there exist a smooth function $A(p)=(A^1(p),...,A^h(p))$ defined on $U$ such that $\bold{\Phi}^*_{A(p)}(f^i)(p)=0$ (for $i=1, \ldots h$), i.e. $\Phi^h_{A^h}(...(\Phi^{1}_{A^{1}}(p))...) \in \tilde{\mathcal{H}}$. In the following we prove that $\mathcal{K}$ is the annihilator of the functions
\begin{equation}\label{equation for K}
K^j(p)=\bold{\Phi}^*_{A(p)}(f^{j})(p), \qquad j>h
\end{equation}
 and, since $K^j$ are independent and adapted with respect to the filtration $\mathcal{G}_k$ for $k$ sufficiently large, $\mathcal{K}$ is a submanifold of $\jinf(M,N)$. \\
We start by proving that  if  $p_0 \in \mathcal{K} \cap U$, then $K^j(p_0)=0$ (for $j>h$).
Indeed,  if $p_0 \in \mathcal{K} \cap U$, the point $p_0$ can be reached starting from $p\in \mathcal{H}$ by means of composition of suitable flows $\Phi^{i}_{a^i}$. On the other hand, for any $p_0 \in \mathcal{K} \cap U$,   there exists  $A(p_0)=(A^1, \ldots, A^h)$  such that
$\Phi^{h}_{A^h}(...\Phi^1_{A^1}(p_0)...) \in \tilde{\mathcal{H}}$. Since $\mathcal{H}\subset \tilde{\mathcal{H}}$ and the transversality condition ensures that equation \eqref{manifold tildeH} have a unique solution, we have $\Phi^{h}_{A^h}(...\Phi^1_{A^1}(p_0)...) \in \mathcal{H}$. Therefore
$$K^j(p_0)=\bold{\Phi}^*_{A(p_0)}(f^j)(p_0)=f^j(\Phi^{h}_{A^h}(...\Phi^1_{A^1}(p_0)...)=0$$
for any $j$ and in particular for $j>h$.
In order to prove the other inclusion
we have to ensure that $p_0$ can be reached starting from a point $p\in \mathcal{H}$ by means of the flows  $\Phi^i_{a^i}$. Given  $p \in \tilde{\mathcal{H}}$ such that $\Phi^h_{A^h}(...(\Phi^{1}_{A^{1}}(p_0))...) =p$, the  definition of $A(p_0)$ ensures that $f^i(p)=0$ for $i=1, \ldots h$ whereas by hypothesis we have
$$K^j(p_0) =\bold{\Phi}^*_{A(p_0)}(f^{j})(p_0)=f^j(p)=0 \qquad  j>h.$$ Hence $f^i(p)=0$ $\forall i \in \mathbb{N}$ and  $p\in \mathcal{H}$.
\hfill\end{proof}

\begin{lemma}\label{lemma_main2}
In the hypotheses and with the notations of Lemma \ref{lemma_main}, $\bar V_{G_j}\in T\mathcal{K}$ and $D_i \in T\mathcal{K}$.
\end{lemma}
\begin{proof}
We recall that a vector field $V \in T\mathcal{K}$ if and only if $V(K^j)=0$, where  $K^j$ are given by \eqref{equation for K}.
Since for any $j$ (with $j>h$) there exists a suitable $k$ such that $f^1, \ldots f^j \in \mathcal{G}_k$,  it is possible to chose as coordinates in $\mathcal{K}\cap U\cap \mathcal{G}_k$ the functions $f^i$ ($i=1, \ldots h$) and some  functions $y^1, \ldots, y^r$ (with $r=\dim(\mathcal{G}_k)-h$) such that $\bar{V}_{G_i}(y^l)=0$. In particular, for any $j>h$, there exists a smooth function $L^j$ such that
$$f^j(p)=L^j(f^1(p),...,f^h(p),y^1(p),...y^r(p)).$$
Since $f^1, \ldots, f^h$ vanish on $\tilde{\mathcal{H}}$ we have
$$K^j(p)=L^j(0,...,0,y^1(p),...,y^r(p)),$$
and so $\bar{V}_{G_i}(K^j)=\bar{V}_{G_i}(L^j(0,...,0,y^1(p),..., y^r(p)))=0$. \\
In order to  prove that $D_i \in T\mathcal{K}$, we consider
$$D^{\alpha}_i=\bold{\Phi}^*_{\alpha}(D_i).$$
By definition, being  $D_i \in T\mathcal{H}$, we have that $D^{\alpha}_i \in T\mathcal{K}$ and,  by Theorem \ref{theorem_flow} (see Appendix), there exist smooth functions $C^i_j(\alpha,p)$ such that
$$D^{\alpha}_i=\sum_j C_i^j(\alpha,p) D_j.$$
Moreover, since $\Phi^i_{a^i}$ are diffeomorphisms,  $\spann\{D^{\alpha}_1,...,D^{\alpha}_m\}$ and $\spann\{D_1,...,D_m\}$ have  the same dimension. Hence the matrix $C^i_j$ is invertible for any $\alpha$, ensuring that $D_i \in T\mathcal{K}$.
\hfill\end{proof}

\begin{remark}\label{remark_main}
The functions $K^j$ defined by \eqref{equation for K} are a set of independent invariants for the vector fields $\bar{V}_{G_i}$. Furthermore,
since $\mathcal{K}$ is finite dimensional, it is possible to add a finite number of functions $z^i$ such that $(z^i,K^j)$ form an adapted coordinate system with respect the filtration $\mathcal{G}_k$ for $k$ sufficiently large.
\end{remark}

\begin{theorem}\label{theorem_main1}
In the hypotheses and with the notations of Lemma \ref{lemma_main}, let $V_F$ be an evolution vector field such that  $V_F \in T\mathcal{H}$, $\dim(\spann\{V_F,V_{G_1},...,V_{G_{h}}\})=h+1$ and
$$[G_i,F]=\mu_i F + \sum_k\lambda_i^k G_k \qquad \mu_i,\lambda_{i}^j \in \mathbb{R}$$
Then $V_F \in T\mathcal{K}$.
\end{theorem}
\begin{proof}
Given the  $(m+h+1)$-dimensional distribution
$$\Delta:=\spann\{D_1,...,D_m,V_{G_1},...,V_{G_h},V_F\}, $$
we have  $\Delta |_{\mathcal{H}}\subseteq T\mathcal{H} \oplus \spann \{ V_{G_1}, \ldots, V_{G_h} \}\subseteq T\mathcal{K}|_{\mathcal{H}}$ and,
by hypothesis,
$[\bar{V}_{G_i},\Delta]\subseteq \Delta$.
If we  prove that
$$\Phi^i_{a*}(\Delta)=\Delta,$$
we have $\Delta|_{\mathcal{K}} \subset T\mathcal{K} $ and, in particular,  $V_F \in T\mathcal{K}$.\\
Considering  the coordinate system $z^i,K^j$ of Remark \ref{remark_main}  we can suppose, possibly relabeling some invariant $z^i$  with $K^j$ for some $j$, that we have exactly $h$ coordinates $z^i$.
Eliminating some element of the form $\partial_{K^j}$, the sequence $V_F,\bar{V}_{G_j},D_k,\partial_{K^l}$ form a basis of $T\jinf(M,N)$ and  for any vector field $X \in T\jinf(M,N)$ there exist suitable functions $b,c^i,d^j,e^l$ depending on $a$ and $p \in U_a$ such that
$$X_a:=\Phi^i_{a*}(X)=b(a,p)V_F+\sum_{j,k,l} c^j(a,p)\bar{V}_{G_j}+d^k(a,p)D_k+e^l(a,p)\partial_{K^l}.$$
From the definition of $X_a$ and using $[\bar{V}_{G_i},\Delta] \subset \Delta$ and $[\bar{V}_{G_i},\partial_{K^l}]\in \Delta$,
we obtain that the functions $e^l$ must solve the equations
$$\partial_a(e^l)=-\bar{V}_{G_i}(e^l).$$
Moreover, since $X_0=X\in \Delta$, $e^l(0,p)=0$ and, from the previous equation, we get $e^l(a,p)=0$ for any $a$, which ensures $X_a \in \Delta$ for any $a$.
\hfill\end{proof}

\begin{remark}
Theorem \ref{theorem_main1} still holds if we consider $r$ functions $F_i \in \mathcal{F}$  such that
$\dim(\spann\{V_{F_1},...,V_{F_r},V_{G_1},...,V_{G_{h}}\})=r+h$, $V_{F_i} \in T\mathcal{H}$ for any $i=1,...,r$ and
$$[G_i,F_j]=\sum_{k,l} (\mu^k_{i,j} F_k+\lambda^l_{i,j}G_l)$$
for some constants  $\lambda^l_{i,j},\mu^k_{i,j} \in \mathbb{R}$.
\end{remark}

\begin{theorem}\label{theorem_main2}
In the hypotheses and with the notations of Lemma \ref{lemma_main}, if $F,G_i$ are real analytic,  $\mathcal{H}$ is defined by real analytic
functions and, denoting by $L=\langle  F, G_1, \ldots, G_h \rangle$ the Lie algebra generated by $F$ and $G_i$, we have
$$L|_{\mathcal{H}}  \subset T\mathcal{H} \oplus \spann\{V_{G_1},...,V_{G_h}\},$$
then $V_F \in T\mathcal{K}$.
\end{theorem}
\begin{proof}
We note that the functions  $K^i$  defined by \eqref{equation for K} are real analytic if the vector fields $\bar{V}_{G_i}$ and  the submanifold $\mathcal{H}$ are real analytic. \\
The vector field $V_F$ is in $T\mathcal{K}$ if for any $p_0 \in \mathcal{K}$ and any $K^i$ we have
$$V_F(K^i)(p_0)=0.$$
We know that if $p_0 \in \mathcal{K}$ there exists an $\alpha=(a^1,...,a^{h}) \in \mathbb{R}^{h}$  and $p_1\in \mathcal{H}$ such that
$$p_0=\Phi^h_{a^h}(...(\Phi^{1}_{a^{1}}(p_1))...).$$
Moreover, being $K^i$ invariants of $ \Phi^j_{a^j}$  we have
\begin{eqnarray*}
V_F(K^i)(p_0)&=&{\bold{\Phi}}^*_{\alpha}(V_F(K^i))(p_1)\\
&=&{\bold{\Phi}}^*_{\alpha}(V_F)(K^i)(p_1).
\end{eqnarray*}
Since the previous expression is real analytic it is sufficient to prove that any derivative of any order with respect $a^i$ evaluated in $(a^1,...,a^{h})=0$ is zero.
It is easy to  verify that
$$\partial^{k_1}_{a^1}(...\partial^{k_{h}}_{a^{h}}(\tilde{\bold{\Phi}}^*_{\alpha}(V_F))...)|_{\alpha=0}=\bar{V}^{k_h}_{G_h}
(...(\bar{V}^{k_1}_{G_1}(V_F))...),$$
where we use the notation
$$\bar{V}^k_{G_i}(X)=\underbrace{[\bar{V}_{G_i},[...[\bar{V}_{G_i},X]...]]}_{k\text{ times}}.$$
By hypothesis $\bar{V}^{k_h}_{G_h}(...(\bar{V}^{k_1}_{G_1}(V_F))...) \in T\mathcal{K}|_{\mathcal{H}}$ and so for any $p_1 \in \mathcal{H}$ and any $K^i$
$$\bar{V}^{k_1}_{G_1}(...(\bar{V}^{k_{h'}}_{G_{h'}}(V_F))...)(K^i)(p_1)=0.$$
\hfill\end{proof}

\begin{remark}
If $\bar{V}_{G_i},V_F$ are real analytic and $\mathcal{H}$ is defined by real analytic equations, Theorem \ref{theorem_main1} implies Theorem \ref{theorem_main2}. On the other hand  Theorem \ref{theorem_main1} turns out to  be very useful  when we consider  smooth (not analytic) invariant manifolds $\mathcal{H}$.
\end{remark}

\begin{remark}
It is important to note that Theorems \ref{theorem_main1} and \ref{theorem_main2} hold also if $\mathcal{H}$ is a manifold with boundary.
In this case if $V_{G_{1}},...,V_{G_{h}} \in T(\partial\mathcal{H} )$  we obtain that $\mathcal{K}$ is also a local manifold with boundary.
\end{remark}

\section{Examples}\label{section_examples}

\subsection{The general algorithm}\label{subsection_general_algorithm}

For the convenience of the reader we start this subsection  describing  the general algorithm  we use in the   examples. Given a PDE of the form
\begin{equation}\label{prototipo_es}
\partial_t(u^i)=c^0(t)F^i(x,u,u_{\sigma})+\sum_{k=1}^h c^k(t) G^i_k(x,u,u_{\sigma}),
\end{equation}
where $G_k \in \mathcal{F}^n$ admit characteristic flows, the first step is to compute the characteristic flows of $G_k$ integrating the characteristic vector fields $\bar{V}_{G_k}$ and obtaining the local diffeomorphisms $\Phi^i_a$. Then,  considering  the  differential constraint $\mathcal{H}$ for $F$  given by the equations $f^i(x,u,u_{\sigma})=0$  and  their differential consequences $D^{\sigma}(f^i)=0$, we choose $h$ independent functions $g^i$ between the functions $f^i,D^{\sigma}(f^i)$  such that the matrix $(\bar{V}_{G_k}(g^i))$ is non-singular and we solve the equations
\begin{equation}\label{equation_algorithm}
\bold{\Phi}_{\alpha}^*(g^i)(x,u,u_{\sigma})=0,
\end{equation}
where $\alpha=(a^1,...,a^h)$, obtaining  $a^i=A^i(x,u,u_{\sigma})$.
Hence, the new differential constraint $\mathcal{K}$ is obtained by replacing $\alpha$ with $(A^1(x,u,u_{\sigma}),...,A^h(x,u,u_{\sigma}))$ in the expression
$$\bold{\Phi}^*_{(A^1(x,u,u_{\sigma}),...,A^h(x,u,u_{\sigma}))}(h^i)=0,$$
where $h^i$ are all the other functions $f^j,D^{\sigma}(f^k)$. \\
We remark that, in order to integrate the system of ODEs representing the evolution equation on $\mathcal{K}$ and the connection $\mathcal{C}|_{\mathcal{K}}$
representing the reduction function $K(x,b)$, we have to compute the coordinate expressions for the vector fields $V_{F},V_{G_k},D_i$ restricted to $\mathcal{K}$. \\
For this purpose we choose a coordinate system
given by $x^i$, some coordinate system $y^i$ on $\mathcal{H}$ and  the functions $a^i=A^i(x,u,u_{\sigma})$.\\
Using coordinates $(x^i,y^j,a^k)$  the vector fields $V_{F},V_{G_k},D_i$ have a rather simple expression. Obviously  $D_i(x^j)=\delta^j_i$, but we have also  $D_i(a^k)|_{\mathcal{K}}=0$.
Furthermore, if the hypotheses of Theorem \ref{theorem_main1} hold,  $V_{G_k}(a^k),V_{F}(a^k)$ depend only on the coordinates $a^k$.
Indeed on $\mathcal{H}$, or equivalently on the submanifold  $a^k=0$ in $\mathcal{K}$, we have  $\partial_{a^k}=-\bar{V}_{G_k}$ (the minus sign owing  to the fact that we use the pull-back in \refeqn{equation_algorithm}). This means that for $\alpha=(a^1,...,a^h)$ we have
$$\partial_{a^k}=\bold{\Phi}^*_{\alpha}(-\bar V_{G_k})$$
and it is easy to prove that $\partial_{a^k}=\sum_j \tilde{C}^j_k(\alpha) \bar{V}_{G_j}$ and so
$$\bar{V}_{G_k}=\sum_j C^j_k(\alpha) \partial_{a^j},$$
where $C$ is the inverse matrix of $\tilde{C}$. \\
Since $\bar{V}_{G_k}=V_{G_k}-\sum_l \tilde{h}^l D_l$ (for some functions $\tilde{h}^l \in \mathcal{F}$) and  $D_i(a^k)=0$,
the expression $V_{G_k}(a^l)$ depends only on $a^1,...,a^h$. The situation is completely similar for $V_F$.\\
\begin{remark}\label{remark_algorithm}
It is important to note that, in the hypotheses of Theorem \ref{theorem_main1},  $V_{G_k}(a^l)$ and  $V_{F}(a^l)$
do not depend on the choice of $\mathcal{H}$ and on the functions $g^i$ but  only on the order
we choose to apply the pull-back in equation \refeqn{equation_algorithm}.
\end{remark}
Once we have  the expressions of $V_F,D_i,V_{G_k}$ in coordinates $(x^i,y^j,a^k)$
\begin{eqnarray*}
V_F&=&\sum_j\phi_0^j \partial_{y^j}+\sum_k \psi_0^k \partial_{a_k}\\
V_{G_h}&=&\sum_j\phi_h^j \partial_{y^j}+\sum_k \psi_h^k \partial_{a_k}\\
D_i&=&\partial_{x^i}+\sum_j\tilde{\phi}_i^j \partial_{y^j}
\end{eqnarray*}
the reduced system for the unknown functions  $Y^i(x,t), \tilde{A}^l(x,t)$ is
\begin{eqnarray*}
\partial_t(Y^i)&=&\sum_{k=0}^h c^k(t) \phi_k^i(Y(x,t),\tilde{A}(x,t))\\
\partial_t(\tilde{A}^l)&=&\sum_{k=0}^h c^k(t) \psi_k^i(\tilde{A}(x,t))\\
\partial_{x^j}(Y^i)&=&\tilde{\phi}_j^i(Y(x,t),\tilde{A}(x,t))\\
\partial_{x^j}(\tilde{A}^l)&=&0.
\end{eqnarray*}

\subsection{A non-linear transport equation with dissipation}\label{subsection_transport_equation}

Let us consider the following equation
\begin{equation}\label{equation_transport}
 \partial_t(u)=c^1(t) u u_x+c^2(t) u.
\end{equation}
Since, for $c^1(t)=1,c^2(t)=0$,  equation \refeqn{equation_transport} is a non-linear transport equation $\partial_t(u)=u u_x$ (see for example \cite{Whitham}),
 if $c^2(t)\leq 0$ and $u\geq 0$ the term $u$ can be considered as a dissipation factor. Furthermore being equation \refeqn{equation_transport}
 for $c^1(t)=1,c^2(t)=0$  the first equation of an Hamiltonian hierarchy (see, e.g.,  \cite{Olver3}), the complete equation \refeqn{equation_transport} can be considered as a non-isospectral perturbation.\\
We can see \eqref{equation_transport} as an evolution PDE of the form \eqref{prototipo_es} with
\begin{eqnarray*}
F&=&0\\
G_1&=&u u_x\\
G_2&=& u.
\end{eqnarray*}
In this case, since $V_F,V_{G_1},V_{G_2}$ are not linearly independent, Theorem \ref{theorem_main1} can not be applied. However Lemma \ref{lemma_main} and \ref{lemma_main2} are enough to provide a differential constraint $\mathcal{K}$ for \eqref{equation_transport}, since $V_F \in T\mathcal{K}$ being $V_F=0$.
The characteristic vector fields for $G_1$ and $G_2$ are  $\bar{V}_{G_1}=V_{G_1}-u D_x$ and  $\bar{V}_{G_2}=V_{G_2}$ with corresponding characteristic flows
\begin{eqnarray*}
\Phi^1_a(x)&=&x-au\\
\Phi^1_a(u)&=&u\\
\Phi^1_a(u_x)&=&\frac{u_x}{1-au_x}\\
\Phi^1_a(u_k)&=&\frac{D_x(\Phi^1_a(u_{k-1}))}{1-au_x}\\
\Phi^2_b(x)&=&x\\
\Phi^2_b(u_k)&=&e^bu_k.
\end{eqnarray*}
In order to construct  a differential constraint for \eqref{equation_transport} we consider the differential constraint $\mathcal{H}$ for $V_F$  defined by $f^1:=u-x^2=0$ and  all  its differential consequences $f^2:=u_x-2x=0$, $f^3:=u_{xx}-2=0$ and $f^k:=u_{k-1}=0$ for $k>2$.
From  equations $\bold{\Phi}^*_{\alpha}(f^1)=\bold{\Phi}^*_{\alpha}(f^2)=0$ we obtain
\begin{eqnarray*}
a&=&-\frac{u_x(xu_x-2u)}{4(xu_x-u)^2}\\
b&=&\log\left(\frac{4(xu_x-u)^2}{uu_x^2}\right)
\end{eqnarray*}
and conversely we have
\begin{eqnarray*}
u&=&\frac{2ax+1-\sqrt{4ax+1}}{2a^2e^b}\\
u_x&=&-\frac{1-\sqrt{4ax+1}}{ae^b\sqrt{4ax+1}}.
\end{eqnarray*}
Using the previous expressions for $a,b$ in $\bold{\Phi}^*_{\alpha}(f^3)=0$ we find that $\mathcal{K}$ is defined by the vanishing  of the function
$$2u^2u_{xx}-xu_x^3+uu_x^2$$
and all its differential consequences.
Choosing   the coordinate system $(x,a,b)$ on $\mathcal{K}$ we have
\begin{eqnarray*}
V_{G_1}&=&-e^{-b}\partial_a\\
V_{G_2}&=&-\partial_b\\
D_x&=&\partial_x.
\end{eqnarray*}
Hence  the function $A(x,t)$ and $B(x,t)$ have to satisfy
\begin{eqnarray*}
\partial_t(A)&=&-c^1(t)e^{-B}\\
\partial_t(B)&=&-c^2(t)\\
\partial_x(A)&=&0\\
\partial_x(B)&=&0
\end{eqnarray*}
and we get
\begin{eqnarray}
A(x,t)&=&A(t)=a_0-\int_0^t{c^1(s)e^{-b_0+\int_0^s{c^2(\tau)d\tau}}ds}\label{equation_a1}\\
B(x,t)&=&B(t)=b_0-\int_0^t{c^2(s)ds}.\label{equation_b1}
\end{eqnarray}
Therefore,  using the expression of $u$ in terms of  $x,a,b$ we obtain
$$U(x,t)=\frac{2A(t)x+1-\sqrt{4A(t)x+1}}{2A(t)^2e^{B(t)}}.$$
Indeed Remark \ref{remark_algorithm} provides a more general result.
In fact,  if we consider $A(t),B(t)$  given by equations \refeqn{equation_a1} and \refeqn{equation_b1}, and we chose  any  manifold $\mathcal{H}$ of the form  $u=H(x)$ (with $H(x)$  a non-linear function) we have that any function $U(x,t)$ solution to
\begin{equation}\label{equation_transport1}
U(x,t)-e^{-B(t)}H(x-A(t)e^{B(t)}U(x,t))=0
\end{equation}
 is a solution to equation \refeqn{equation_transport}.
Hence in this case  the use of the implicit form for the constraint turns out to be  more effective than the direct use of the reduction function $K$.

\subsection{Modified heat equation}\label{subsection_heat_equation}

Let us consider the following equation
$$u_t=c^0u_{xx}+c^1(t) xu+c^2(t) x^2 u,$$
which can be seen as an equation of the form \eqref{prototipo_es} with
\begin{eqnarray*}
F&=&u_{xx}\\
G_1&=&xu\\
G_2&=&x^2u.
\end{eqnarray*}
This equation (with $c^1,c^2$ constants) has already been studied  in \cite{Rosencrans,Steinberg} and coincides with the Zakai equation of the simplest Kalman filter in one dimension if $c^1=1/2$ and $c^2$
is the derivative of a Brownian motion process (see \cite{Bain}).\\
The vector fields  $V_{G_1},V_{G_2}$ form an abelian Lie algebra and admit strong characteristics. Furthermore  $F$, $G_1$, $G_2$, $\tilde{G}_1=u_x$,
$\tilde{G}_2=xu_x$, $\tilde{G}_3=u$ form a Lie algebra. \\
Let $\mathcal{H}$ be the submanifold of $\jinf(\mathbb{R},\mathbb{R})$ defined by  $f^1:=u_x=0$ and  all its differential consequences $f^{k+1}:=D^k_x(f^1)=0$ for any $k \in \mathbb{N}$. It is easy to prove that $V_F,V_{\tilde {G}_i} \in T\mathcal{H}$ for any $i=1,2,3$. In this situation $F,G_1,G_2$  satisfy the hypotheses of Theorem \ref{theorem_main2} on $\mathcal{H}$
and, in order to  find the equations of  $\mathcal{K}$, we start by computing
the characteristic flows of $G_1$ and $G_2$
\begin{eqnarray*}
&\Phi^{1*}_a(x)=\Phi^{2*}_b(x)=x&\\
&\Phi^{1*}_a(u)=e^{ax}u&\\
&\Phi^{2*}_b(u)=e^{bx^2}u.&
\end{eqnarray*}
By Theorem \ref{theorem_flow} we obtain the characteristic flows for the derivatives of an  order
\begin{eqnarray*}
\bold{\Phi}^*_{\alpha}(u_x)&=&e^{ax+bx^2}((a+2bx)u+u_x)\\
\bold{\Phi}^*_{\alpha}(u_{xx})&=&e^{ax+bx^2}(((a+2bx)^2+2b)u+2(a+2bx)u_x+u_{xx}).
\end{eqnarray*}
and, using  equations $\bold{\Phi}^*_{\alpha}(u_x)=\bold{\Phi}^*_{\alpha}(u_{xx})=0$, we get
\begin{eqnarray*}
u_{x}&=&-(a+2bx)u\\
u_{xx}&=&u((2bx+a)^2-2b).
\end{eqnarray*}
Therefore we can  express  $a,b$ as functions of $u_x,u_{xx}$ as follows
\begin{eqnarray*}
a&=&\frac{(uu_{xx}-u_x^2)x-uu_x}{u^2}\\
b&=&\frac{u_x^2-uu_{xx}}{2u^2}
\end{eqnarray*}
and equation $\Phi_{\alpha}(u_{xxx})=0$ defining (together with all its differential consequences) the manifold $\mathcal{K}$ is
\begin{eqnarray*}
u_{xxx}&=&\frac{3u_xu_{xx}}{u}-\frac{2u_x^3}{u^2}.\\
\end{eqnarray*}
Since the manifold $\mathcal{K}$ is  four dimensional,  we use  coordinates $(x,u,a,b)$ on it and,  computing
the components of the vector fields $V_{F_i}$, $V_{G_j}$ and $D_x$,  we get
\begin{eqnarray*}
V_{F}&=&u((2bx+a)^2-2b)\partial_u-4ab\partial_a-4b^2\partial_b\\
V_{G_1}&=&xu\partial_u-\partial_a\\
V_{G_2}&=&x^2u\partial_u-\partial_b\\
D_x&=&\partial_x-u(a+2bx)\partial_u.
\end{eqnarray*}
Hence the equations for $U,A,B$ are
\begin{eqnarray*}
\partial_t(U)&=&U[c^0((2Bx+A)^2-2B)+c^1(t)x+c^2(t)x^2]\\
\partial_t(A)&=&-4c^0AB-c^1(t)\\
\partial_t(B)&=&-4c^0B^2-c^2(t)\\
\partial_x(U)&=&-U(A+2Bx)\\
\partial_x(A)&=&0\\
\partial_x(B)&=&0.\\
\end{eqnarray*}
The previous system has a unique solution such that $U(x_0,t_0)=u_0,A(x_0,t_0)=a_0,B(x_0,t_0)=b_0$. Without loss of generality we can suppose that $x_0=0$ and, in order to simplify computation,  we also suppose  $c^1,c^2 \in \mathbb{R}$. The equations for $A,B,U$ in $t$ derivative form a triangular system (linear in $A,U$ and with Riccati form not dependent on time in $B$) and can be solved explicitly getting some functions $U^0(t)=U(0,t), A^0(t)=A(0,t), B^0(t)=B(0,t)$ satisfying $U^0(t_0)=u_0, A^0(t_0)=a_0$ and $B^0(t_0)=b_0$. \\
Hence we can explicitly integrate the equations for $x$  and we get
\begin{eqnarray*}
A(x,t)&=&A^0(t)\\
B(x,t)&=&B^0(t)\\
U(x,t)&=&e^{-A^0(t)x-B^0(t)x^2}U^0(t).
\end{eqnarray*}
The function $U(x,t)$ is the well known Gaussian solution for the modified heat equation.

\subsection{An integrable two dimensional system}\label{subsection_integrable}

Let us consider the following system
\begin{eqnarray*}
u_t&=&c^0u_{xx}+\sum_{i,j \leq N}c^{i,j}(t)v_iv_j\\
v_t&=&v_{xx},
\end{eqnarray*}
where $v_{k}=D^k_x(v)$ for $k>0$ and $v_0=v$. In the case $c^{i,j}=0$ for $i,j\geq 2$ and $c^{i,j}=cost$ for $i,j<2$,
this system admits an infinite number of higher order symmetries and a recursion operator (see for example \cite{Beukers}). Furthermore denoting by
\begin{eqnarray*}
F&=&\left(\begin{array}{c}
c^0 u_{xx}\\
v_{xx}\end{array} \right),\\
G_{i,j}&=&\left(\begin{array}{c}
v_i v_j \\
0 \end{array}\right),
\end{eqnarray*}
it is easy to prove that $V_{F},V_{G_{i,j}}$ form a pro-finite Lie algebra so that  this provides a toy-model for the pro-finite Lie algebras
used in non-linear filtering problem (see \cite{Hazewinkel}).
If we consider the submanifold $\tilde{\mathcal{H}}\subset \jinf(\mathbb{R},\mathbb{R}^2)$ given by the equation
$$g=v_k-\sum_{i<k} d_i v_i=0 \qquad \qquad d_i \in \mathbb{R}$$
 and its differential consequences, obviously $D_x \in T\tilde{\mathcal{H}}$. Furthermore we have  $V_F, V_{G_{i,j}} \in T\tilde{\mathcal{H}}$
 and $L_1=\spann\{V_{G_{i,j}}\}$  restricted on $\tilde{\mathcal{H}}$ is finite dimensional. So putting $\tilde{L}=\spann\{V_{G_{i,j}}|i,j<k\}$,
 if $\mathcal{H}$ is a finite dimensional submanifold of $\tilde{\mathcal{H}}$ such that $\tilde{L}$ has maximal rank on $\mathcal{H}$,
 $D_x \in T\mathcal{H}$ and $V_F \in T\mathcal{H}$.\\
The hypotheses of Theorem \ref{theorem_main2} are satisfied: indeed, denoting by  $L=\spann\{V_F,V_{G_{i,j}}\}$, we have
$$[L,L]|_{\tilde{\mathcal{H}}}=L_1|_{\tilde{\mathcal{H}}}=\tilde{L}|_{\tilde{\mathcal{H}}}.$$
Hereafter, in order to simplify computation, we take $c^{i,j}=0$ for $i,j>1$, $c^0=1$ and $g=v_{xx}-\beta v$.
In this case we choose as submanifold $\mathcal{H}$ of $\tilde{\mathcal{H}}$ the set of zeros of $h=u_x-\gamma u$. Hence, writing
$V_1=V_{G_{0,0}},V_2=V_{G_{1,1}},V_3=V_{G_{0,1}}$ and denoting by $\Phi^i$ the corresponding characteristic flows, we have
\begin{eqnarray*}
\Phi^1_a\left(\begin{array}{c}
u\\
v\end{array}\right)&=&\left(\begin{array}{c}
u+av^2\\
v\end{array}\right),\\
\Phi^2_b\left(\begin{array}{c}
u\\
v\end{array}\right)&=&\left(\begin{array}{c}
u+bv_x^2\\
v\end{array}\right),\\
\Phi^3_c\left(\begin{array}{c}
u\\
v\end{array}\right)&=&\left(\begin{array}{c}
u+cvv_x\\
v\end{array}\right).
\end{eqnarray*}
So, from  $\bold{\Phi}_{\alpha}^*(h)=\bold{\Phi}_{\alpha}^*(D_x(h))=\bold{\Phi}_{\alpha}^*(D^2_x(h))=0$, we obtain that $\mathcal{K}$ is defined by
$$u_{4}=(u_{xxx}-4\beta u_x)\gamma+4\beta u_{xx}.$$
On $\mathcal{K}$ we use the natural coordinate system $(x,u,v,v_x,a,b,c)$, where
\begin{eqnarray*}
a&=&\frac{((2\beta u-u_{xx})v_x^2+2\beta u_x vv_x-2\beta^2uv^2)\gamma^3+(u_{xxx}-4\beta u_x)v_x^2\gamma^2+
}{
(2v_x^4-4\beta v^2v_x^2+2\beta^2v^4)\gamma^3+(-8\beta v_x^4+16\beta^2v^2v_x^2-8\beta^3v^4)\gamma}\\
&&+\frac{((4\beta u_{xx}-8\beta^2u)v_x^2-2\beta u_{xxx}vv_x+8\beta^3uv^2)\gamma+(8\beta^2u_x-2\beta u_{xxx})v_x^2}{
(2v_x^4-4\beta v^2v_x^2+2\beta^2v^4)\gamma^3+(-8\beta v_x^4+16\beta^2v^2v_x^2-8\beta^3v^4)\gamma}\\
&&+\frac{(2\beta^2u_{xxx}-8\beta^3u_x)v^2}{(2v_x^4-4\beta v^2v_x^2+2\beta^2v^4)\gamma^3+(-8\beta v_x^4+16\beta^2v^2v_x^2-8\beta^3v^4)\gamma}\\
b&=&-\frac{
(2uv_x^2-2u_xvv_x+(u_{xx}-2\beta u)v^2)\gamma^3+(4\beta u_x-u_{xxx})v^2\gamma^2}{(2v_x^4-4\beta v^2 v_x^2+2\beta^2 v^4)\gamma^3+(-8\beta v_x^4+16\beta^2 v^2v_x^2-8\beta^3v^4)\gamma}\\
&&-\frac{(-8\beta u v_x^2+2u_{xxx}vv_x+(8\beta^2u-4\beta u_{xx})v^2)\gamma+(8\beta u_x-2 u_{xxx})v_x^2}{
(2v_x^4-4\beta v^2 v_x^2+2\beta^2 v^4)\gamma^3+(-8\beta v_x^4+16\beta^2 v^2v_x^2-8\beta^3v^4)\gamma}\\
&&-\frac{+(2\beta u_{xxx}-8\beta^2u_x)v^2}{
(2v_x^4-4\beta v^2 v_x^2+2\beta^2 v^4)\gamma^3+(-8\beta v_x^4+16\beta^2 v^2v_x^2-8\beta^3v^4)\gamma}\\
c&=&-\frac{(u_xv_x^2-u_{xx}vv_x+\beta u_xv^2)\gamma^2+(u_{xxx}-4\beta u_x)v v_x\gamma-u_{xxx}v_x^2}{(v_x^4-2\beta v^2 v_x^2+\beta^2v^4)\gamma^2-4\beta v_x^4+8\beta^2v^2v_x^2-4\beta^3v^4}\\
&&-\frac{4\beta u_{xx}vv_x-\beta u_{xxx}v^2}{(v_x^4-2\beta v^2 v_x^2+\beta^2v^4)\gamma^2-4\beta v_x^4+8\beta^2v^2v_x^2-4\beta^3v^4}
\end{eqnarray*}
In this coordinate system we have
\begin{eqnarray*}
V_F&=&\left((bv_x^2+cvv_x+av^2+u)\gamma^2-(2b\beta+2a)v_x^2-4\beta cvv_x-(2b\beta^2+2a\beta)v^2\right)\partial_u\\
&&+\beta v\partial_v+\beta v_x\partial_{v_x}+2\beta^2b\partial_a+2a\partial_b+2\beta c \partial_c\\
V_1&=&v^2\partial_u-\partial_a\\
V_2&=&v_x^2\partial_u-\partial_b\\
V_3&=&vv_x\partial_u-\partial_c\\
D_x&=&\partial_x+\left((bv_x^2+cvv_x+av^2+u)\gamma-cv_x^2-(2b\beta+2a)vv_x-\beta cv^2\right)\partial_u\\
&&+ v_x\partial_{v}+\beta v\partial_{v_x}.
\end{eqnarray*}
Fixing $(t_0,x_0)$ and the initial conditions $U(t_0,x_0)=u_0,V(t_0,x_0)=v_0...$, the functions $U,V, V_x,A,B,C$ must satisfy the following overdetermined system of equations
\begin{eqnarray*}
\partial_t(A)&=&2\beta^2B-c^{0,0}(t)\\
\partial_t(B)&=&2A-c^{1,1}(t)\\
\partial_t(C)&=&2 \beta C-c^{0,1}(t)\\
\partial_t(U)&=&\gamma^2 U+(BV_x^2+CVV_x+AV^2)\gamma^2-(2B\beta+2A)V_x^2-4\beta CVV_x\\
&&-(2B\beta^2+2A\beta)V^2+c^{0,0}(t)V^2+c_{1,1}(t)V_x^2+c^{0,1}(t)VV_x\\
\partial_t(V)&=&\beta V\\
\partial_t(V_x)&=& \beta V_x\\
\partial_x(A)&=&0\\
\partial_x(B)&=&0\\
\partial_x(C)&=&0\\
\partial_x(U)&=&\gamma U +(BV_x^2+CVV_x+AV^2)\gamma-CV_x^2-(2B\beta+2A)VV_x-\beta CV^2\\
\partial_x(V)&=& V_x\\
\partial_x(V_x)&=&\beta V.
\end{eqnarray*}
In the part of system with the $t$ derivative, the equations for  $A,B,C$  do not depend on the other variables and  are linear and non-homogeneous with respect to $A,B,C$. So, considering the matrix
$$
S(t)=\left(\begin{array}{ccc}
\frac{1}{2}\cosh(2\beta t) & \frac{\beta}{2} \sinh(2 \beta t) & 0\\
\frac{1}{2\beta} \sinh(2 \beta t) & \frac{1}{2} \cosh(2 \beta t) & 0\\
0& 0 &e^{2 \beta t}
\end{array}\right),
$$
we have
$$
\left(\begin{array}{c}
A^0(t)\\
B^0(t)\\
C^0(t)
\end{array}\right)=S(t-t_0)\cdot\left(\begin{array}{c}
a_0\\
b_0\\
c_0
\end{array} \right)+S(t-t_0)\int_{t_0}^t{S(-s+t_0)\cdot\left(\begin{array}{c}
c^{0,0}(s)\\
c^{1,1}(s)\\
c^{0,1}(s)\end{array} \right)ds},
$$
where $A^0(t)=A(x_0,t),etc.$. Moreover, since the equations in $t$ for $V,V_x$ are linear, we have
\begin{eqnarray*}
V^0(t)&=&v_0 e^{\beta (t-t_0)}\\
V^0_{x}(t)&=&v_{x,0} e^{\beta(t-t_0)}
\end{eqnarray*}
and, being also the equation for $U$  linear in $U$ and depending on $v^0,v^0_x,...$,   we obtain
\begin{eqnarray*}
U^0(t)&=&u_0 e^{\gamma^2 (t-t_0)}+e^{\gamma^2(t-t_0)}\left(\int_{t_0}^t{e^{-\gamma^2(s-t_0)}\gamma^2(B^0(s)V^0_x(s)^2+C^0(s)V^0(s)V^0_x(s))ds}\right.\\
&&+\left.\int_{t_0}^t{e^{\gamma^2(s-t_0)}(A^0(s)V^0(s)^2\gamma^2+(-2\beta B^0(s)-2A^0(s))V^0_x(s)^2)ds}\right.\\
&&+\left.\int_{t_0}^{t}{e^{-\gamma^2(s-t_0)}((-4\beta C^0(s)V^0(s)V^0_x(s)-2\beta^2B^0(s)-2\beta A^0(s))V^0(s)^2)ds}\right.\\
&&+\left.\int_{t_0}^t{e^{\gamma^2(s-t_0)}(c_{0,0}(s)V^0(s)^2+c_{1,1}(s)V^0_x(s)^2+c_{0,1}(s)V^0(s)V^0_x(s))ds}\right).
\end{eqnarray*}
Finally, integrating the equations for $x$  we get
\begin{eqnarray*}
A(x,t)&=&A^0(t)\\
B(x,t)&=&B^0(t)\\
C(x,t)&=&C^0(t)\\
V(x,t)&=&\frac{V^0_x(t)}{\sqrt{\beta}}\sinh(\sqrt{\beta} (x-x_0))+V^0(t)\cosh(\sqrt{\beta}(x-x_0))\\
V_x(x,t)&=&\sqrt{\beta} V^0(t) \sinh(\sqrt{\beta} (x-x_0))+ V^0_x(t) \cosh(\sqrt{\beta}(x-x_0))\\
U(x,t)&=&U^0(t)e^{\gamma (x-x_0)}+e^{\gamma(x-x_0)}\left(\int_{x_0}^x{e^{-\gamma(y-x_0)}B^0(t)V_x(y,t)^2\gamma dy}\right.\\
&&\int_{x_0}^x{e^{-\gamma(y-x_0)}((C^0(t)V(y,t)V_x(y,t)+a^0(t)V(y,t)^2-C^0(t)V_x(y,t)^2)\gamma dy}\\
&&+\left.\int_{x_0}^x{e^{-\gamma(y-x_0)}((-2\beta B^0(t)-2A^0(t))V(y,t)V_x(y,t)-\beta c^0(t)V(y,t)^2)dy}\right).
\end{eqnarray*}

\subsection{Perturbed KdV equation}\label{subsection_KdV}

Let us consider the following equation
\begin{equation}\label{equationKdV}
\partial_t(u)=(u_{xxx}+uu_x)+c^1(t)+c^2(t)(xu_x+2u),
\end{equation}
corresponding to \eqref{prototipo_es} with
\begin{eqnarray*}
F&=&u_{xxx}+uu_x\\
G_1&=&1\\
G_2&=&xu_x+2u.
\end{eqnarray*}
If $c^2=0$ and $c^1$  is the derivative of a Brownian motion,  equation \eqref{equationKdV} can be seen as a stochastic perturbation of KdV equation (see \cite{Wadati,Xie}) whereas in all the other cases  \eqref{equationKdV} can be interpreted as  a non-isospectral perturbation of KdV equation (see e.g.\cite{Calogero,Gordoa_Pickering_Wattis}).
As submanifold $\mathcal{H}$ we consider the annihilator of $g=u_{xx}+\frac{1}{2}u^2-\beta_0 u$ (where $\beta_0 \in \mathbb{R}_+$) which contains the one soliton solution to the KdV equation with velocity $\beta_0$. \\
The flows of $V_1=V_{G_1}$ and of $\bar{V}_2=V_{G_2}-xD_x$ are given by
\begin{eqnarray*}
\Phi^1_{a}(u)&=&u+a\\
\Phi^2_{b}(u_k)&=& e^{(k+2)b} u_k,
\end{eqnarray*}
where $u_k=D^k(u)$ and $u_0=u$. Hence, solving $\bold{\Phi}^*_{\alpha}(g)=\bold{\Phi}^*_{\alpha}(D_x(g))=0$, we obtain
\begin{eqnarray*}
b=\frac{1}{4}\log\left(\frac{\beta^2_0u^2_x}{u_{xxx}^2+2u_x^2u_{xx}}\right)\\
a=\sqrt{\frac{u_{xxx}^2+2u_x^2u_{xx}}{u_x^2}}-\frac{u_{xxx}}{u_x}-u\\
\end{eqnarray*}
and  $\mathcal{K}$ is given by the zero set of
$$u_{xxxx}-\frac{(u_{xxx}u_{xx}-u_x^3)}{u_x}.$$
In order to simplify  computation we introduce the  coordinate system $(x,\tilde{u},a,\beta,\gamma)$ on $\mathcal{K}$, where
\begin{eqnarray*}
\beta&=&e^{2b}\\
\gamma&=&\frac{1}{2}u_x^2+\frac{1}{6}u^3-\frac{1}{2}\left(\frac{\beta_0}{\beta}-a\right)u^2-\left(\frac{a\beta_0}{\beta}-\frac{a^2}{2}\right)u\\
\tilde{u}&=&\left\{\begin{array}{ccc}
\int_{d}^{u}{\frac{1}{\sqrt{2\left(\gamma-\frac{1}{6}z^3+\frac{1}{2}\left(\frac{\beta_0}{\beta}-a\right)z^2+\left( \frac{a\beta_0}{\beta}-\frac{a^2}{2}\right)  z\right)}}dz}&\text{ if }&u_x>0\\
-\int_{d}^{u}{\frac{1}{\sqrt{2\left(\gamma-\frac{1}{6}z^3+\frac{1}{2}\left(\frac{\beta_0}{\beta}-a\right)z^2+\left( \frac{a\beta_0}{\beta}-\frac{a^2}{2}\right)  z\right)}}dz}&\text{ if }&u_x<0
\end{array}\right.
\end{eqnarray*}
(here  $d \in \mathbb{R}$ is such that $\gamma-\frac{1}{6}d^3+\frac12(\frac{\beta_0}{\beta}-a)d^2+\left( \frac{a\beta_0}{\beta}-\frac{a^2}{2}\right)  d >0$).
Using this coordinate system  it is easy to verify that
\begin{eqnarray*}
V_F&=&\left(\frac{\beta_0}{\beta}-a\right)\partial_{\tilde{u}}\\
V_{G_1}&=&\left(\pm\frac{1}{\sqrt{\gamma-\frac{1}{6}d^3+\frac{1}{2}\left(\frac{\beta_0}{\beta}-a\right)d^2+\left( \frac{a\beta_0}{\beta}-\frac{a^2}{2}\right)  d}} \right) \partial_{\tilde{u}}-\partial_{a}-\left(\frac{a\beta_0}{\beta}-\frac{a^2}{2}\right)\partial_{\gamma}\\
V_{G_2}&=&\left(x\pm \frac{2d}{\sqrt{\gamma-\frac{1}{6}d^3+\frac{1}{2}\left(\frac{\beta_0}{\beta}-a\right)d^2+\left( \frac{a\beta_0}{\beta}-\frac{a^2}{2}\right) d}}-\tilde{u}\right)\partial_{\tilde{u}}+2 a \partial_{a}-2 \beta \partial_{\beta}+6 \gamma\partial_{\gamma}\\
D_x&=&\partial_x+\partial_{\tilde{u}},
\end{eqnarray*}
where in $V_{G_1}$ and $V_{G_2}$ we choose the plus sing if $u_x>0$ and the minus sign if $u_x<0$.\\
The equations for $\tilde{U},A,B,\Gamma$ are
\begin{eqnarray*}
\partial_t(\tilde{U})&=&-c^2(t)\tilde{U}+\left(\frac{\beta_0}{B}-A\right)+\\
&&\pm\frac{c^1(t)+2dc^2(t)}{\sqrt{\Gamma-\frac{1}{6}d^3+\frac{1}{2}\left(\frac{\beta_0}{B}-A\right)d^2+\left( \frac{A\beta_0}{B}-\frac{A^2}{2}\right)  d}}\\
\partial_t(B)&=&-2c^2(t)B\\
\partial_t(A)&=&2c^2(t)A-c^1(t)\\
\partial_t(\Gamma)&=&6c^2(t)\Gamma-\left(\frac{A\beta_0}{B}-\frac{A^2}{2}\right)c^1(t)\\
\partial_x(\tilde{U})&=&1\\
\partial_x(B)&=&0\\
\partial_x(A)&=&0\\
\partial_x(\Gamma)&=&0
\end{eqnarray*}
and this system can be solved as in the previous example. \\
If we consider the particular case $x_0=0$ and $A(0,t_0)=0, B(0,t_0)=1, \Gamma(0,t_0)=0$  we can explicitly compute
\begin{eqnarray*}
B(x,t)&=&B^0(t)=e^{-2C^2(t)}\\
A(x,t)&=&A^0(t)=-(B^0(t))^{-1}\left(\int_{t_0}^t{B^0(s)c^1(s)ds}\right)\\
\Gamma(x,t)&=&\Gamma^0(t)=\frac 12 \beta_0 (B^0(t))^{-1}(A^0(t))^2- \frac 16(A^0(t))^3 ,
\end{eqnarray*}
where $C^2(t)=\int_{t_0}^t{c^2(s)ds}$.
Expressing $\tilde{U}$ as a function of $A^0(t),B^0(t),\Gamma^0(t),u$  we get
\begin{eqnarray*}
\tilde{U}(A^0(t),B^0(t),\Gamma^0(t),u)&=&\pm\int_d^{u}{\frac{dz}{\sqrt{2\left(-\frac{(z+A^0(t))^3}{6}+\frac{\beta_0}{B^0(t)}\frac{(z+A^0(t))^2}{2}\right)}}}\\
&=&\mp 2\sqrt{\frac{B^0(t)}{\beta_0}}\left(\acosh\left(\sqrt{\frac{3\beta_0}{B^0(t)(u+A^0(t))}}\right)\right.\\
&&\left.-\acosh\left(\sqrt{\frac{3\beta_0}{B^0(t)(d+A^0(t))}}\right)\right)
\end{eqnarray*}
and  we obtain
\begin{eqnarray*}
U(x,t)&=&3\beta_0e^{2C^2(t)}\left(\cosh\left(\mp\frac{\sqrt{\beta_0}e^{C^2(t)}\tilde{U}(x,t)}{2}\right.\right.\\
&&\left.\left.+\acosh\left(\sqrt{\frac{3\beta_0}{B^0(t)(d+A^0(t))}}\right)\right)\right)^{-2}-A^0(t).
\end{eqnarray*}
If we solve the equations for $\tilde{U}(x,t)$ we find
$$\tilde{U}(x,t)=x+e^{-C^2(t)}\left(\int_0^t{\left(\frac{\beta_0}{B^0(s)}-A^0(s)\right)ds}\pm\frac{2}{\sqrt{\beta_0}}\acosh\left(\sqrt{\frac{3\beta_0}{B^0(t)(d+A^0(t))}}\right)\right)$$
and we have
\begin{eqnarray*}
U(x,t)&=&3\beta_0e^{2C^2(t)}\left(\cosh\left(\frac{\sqrt{\beta_0}}{2}\left( e^{C^2(t)} x +\int_{t_0}^t{\beta_0e^{3C^2(s)}ds}\right.\right.\right.\\
&&\left.\left.\left. +\int_{t_0}^t{e^{3C^2(s)}\left(\int_{t_0}^s{e^{-2C^2(\tau)}}c^1(\tau)d\tau\right) ds}\right)\right)\right)^{-2}+\\
&&+e^{2C^2(t)}\int_{t_0}^t{e^{-2C^2(s)}c^1(s)ds},
\end{eqnarray*}
where we use the parity of the function $\cosh$ to eliminate $\mp$ sign.

\section{Appendix}\label{appendix}

In this section we discuss the behavior of the Cartan distribution $\mathcal{C}$ under the action of the characteristic flow $\Phi_a$ associated with an evolution vector field $V_G$. An important consequence of the following Theorem  is that $\Phi^*_a(u^i_{\sigma})$ is a polynomial function with respect the variable $u^k_{\sigma'}$ if $|\sigma'|$ is sufficiently large.
\begin{theorem}\label{theorem_flow}
Let $V_G$ be an evolution vector field admitting characteristics  and let $\Phi_a$ be the corresponding  characteristic flow. Denoting by $A$ the $n \times n$ matrix
$$A=(A^j_i):=(D_i(\Phi^*_a(x^j)))|_{i,j},$$
and  by $B=(B^i_j)$ the inverse matrix of $A$, then
\begin{equation}\label{equation_flow1}
\Phi^*_a(D_i)=\sum_j B^j_i D_j
\end{equation}
and,  for any $f \in \mathcal{F}$, we have
\begin{equation}\label{equation_flow2}
\Phi^*_a(D_i(f))=\sum_j B^j_i D_j(\Phi^*_a(f)).
\end{equation}
\end{theorem}
\begin{proof}
Let $\tilde{V}_G= V_G-\sum_i h^i D_i$  be the characteristic vector field of $V_G$. Since
$$[\tilde{V}_G,D_i]=\sum_j D_i(h^j) D_j,$$
the vector field $D^a_i=\Phi^*_a(D_i)$ solves the  equation
\begin{equation}\label{equation_flow3}
\partial_a(D^a_i)=\Phi^*_a([\tilde{V}_G, D_i]=\sum_j D^a_i(\Phi^*_a(h^j)) D^a_j.
\end{equation}
In order to prove \eqref{equation_flow1} we show  that the vector field $\tilde{D}^a_i:=\sum_j B^j_i D_j$ solves  equation \eqref{equation_flow3} as well.
We start by computing
$$\partial_a(A^j_i)=\partial_a(D_i(\Phi^*_a(x^j)))=D_i(\Phi^*_a(\tilde{V}_G(x^j)))=
-D_i(\Phi^*_a(h^j)).$$
Since $B=A^{-1}$  the formula for derivative of the inverse matrix gives
$$\partial_a(B)=-B \cdot \partial_a(A) \cdot B.$$
This means that
$$\partial_a(B^i_j)=\sum_{k,r}B^k_j (D_k(\Phi^*_a(h^r))) B_r^i$$
and we get
\begin{eqnarray*}
\partial_a(\tilde{D}^a_j)&=&\sum_i \partial_a(B^i_j) D_i\\
&=&\sum_{i,k,r} B^k_j (D_k(\Phi^*_a(h^r))) B_r^i D_i \\
&=&\sum_r \tilde{D}^a_j(\Phi^*_a(h^r)) \tilde{D}^a_r.
\end{eqnarray*}
Hence both  $\tilde{D}^a_i$ and $D^a_i$ satisfy equation \eqref{equation_flow3} and we have  $\tilde{D}^a_i=D^a_i$.
\hfill \end{proof}

\begin{remark}\label{remark_flow1}
It is important to note that equation \eqref{equation_flow1} holds in all $\jinf(M,N)$ while equation \eqref{equation_flow2}
holds in $\mathcal{F}_k$ for $k$ sufficiently large.
\end{remark}

\begin{corollary}\label{corollary_flow1}
Given an evolution vector field $V_G$ with corresponding characteristic flow $\Phi_a$,
the expression $\Phi^*_a(u^i_{\sigma})$ is a polynomial function with respect the variable $u^k_{\sigma'}$ if $|\sigma'|$ is sufficiently large.
\end{corollary}
\begin{proof}
If we apply Theorem \ref{theorem_flow} to $f=u^i$ we get $\Phi^*_a(D_k(u^i))= \sum_j B^j_k D_j(\Phi^*_a(u^i))$. Since $D_j(\Phi^*_a(u^i))$ is a linear function with respect the variable $u^i_{\sigma'}$ if $|\sigma'|$ is sufficiently large and $B^j_k \in \mathcal{F}_h$ for some $h \in \mathbb{N}$, applying iteratively Theorem \ref{theorem_flow} we obtain the thesis.
\hfill \end{proof}

\section*{Acknowledgements}
This work was  supported by National Group of Mathematical Physics (GNFM-INdAM).

\bibliographystyle{plain}
\bibliography{invariant_manifold_characteristics}

\end{document}